\documentclass[preprint,12pt]{elsarticle}

\usepackage{bm}
\newcommand{\KL}{\textup{KL}}
\newcommand{\LF}{\textup{LF}}
\newcommand{\HF}{\textup{HF}}
\newcommand{\BF}{\textup{BF}}
\newcommand{\AL}{\mathcal{A}_L}
\newcommand{\AH}{\mathcal{A}_H}

\newcommand{\ExH}{P_f}
\newcommand{\ExL}{P_f^\LF}

\usepackage{amsfonts}
\usepackage{amsthm,physics}
\usepackage{amsmath}
\usepackage{amscd}
\usepackage{t1enc}
\usepackage[mathscr]{eucal}
\usepackage{indentfirst}
\usepackage{graphicx}
\usepackage{graphics}
\usepackage{pict2e}
\usepackage{epic}
\numberwithin{equation}{section}
\usepackage{epstopdf}
\usepackage{amsmath,amsthm,verbatim,amssymb,amsfonts,amscd,graphicx}
\usepackage{mathtools}
\usepackage[dvipsnames]{xcolor}
\usepackage{bm}
\usepackage{soul}

\usepackage[dvipsnames]{xcolor}
\usepackage[margin=3cm]{caption}
\usepackage{amsmath}
\usepackage{amsthm}
\usepackage{amssymb}
\usepackage{amscd}
\usepackage{mathtools}
\usepackage{hyperref}
\usepackage{subcaption}
\usepackage{stmaryrd}
\usepackage{booktabs}
\usepackage{bigints}
\usepackage{mdframed}
\usepackage{upgreek}
\usepackage{mathtools} \mathtoolsset{showonlyrefs=true} \MakeRobust{\eqref}
\usepackage{listings}
\usepackage{color}
\usepackage{fullpage}
\usepackage{stackrel}
\usepackage{textcomp}
\usepackage{enumerate}
\usepackage{accents}
\usepackage{graphicx}
\usepackage{epstopdf}
\usepackage{placeins}
\usepackage{bbm}
\usepackage{bm}
\usepackage{accents}
\usepackage[mathscr]{eucal}
\usepackage[linesnumbered,boxed,titlenumbered,ruled,nofillcomment]{algorithm2e}
\usepackage[T1]{fontenc}
\usepackage{lmodern}
\usepackage{scalerel}
\usepackage{wrapfig}
\usepackage{algorithmic}
\usepackage{float}
\usepackage{soul}

\definecolor{lbcolor}{rgb}{0.9,0.9,0.9}
\theoremstyle{plain}
\newtheorem{theorem}{Theorem} \numberwithin{theorem}{section}

\newtheorem{lemma}[theorem]{Lemma}

\theoremstyle{definition}
\newtheorem{remark}[theorem]{Remark}

\definecolor{mydarkblue}{rgb}{0,0.08,0.45}
\hypersetup{ %
	colorlinks=false,
	linkcolor=mydarkblue,
	citecolor=mydarkblue,
	filecolor=mydarkblue,
	urlcolor=mydarkblue}

\DeclareMathOperator*{\argmin}{arg\,min}

\definecolor{mypink1}{rgb}{0.858, 0.188, 0.478}
\definecolor{mypink2}{RGB}{219, 48, 122}
\definecolor{mypink3}{cmyk}{0, 0.7808, 0.4429, 0.1412}
\definecolor{mygray}{gray}{0.6}

\newlength{\leftstackrelawd}
\newlength{\leftstackrelbwd}
\def\leftstackrel#1#2{\settowidth{\leftstackrelawd}%
{${{}^{#1}}$}\settowidth{\leftstackrelbwd}{$#2$}%
\addtolength{\leftstackrelawd}{-\leftstackrelbwd}%
\leavevmode\ifthenelse{\lengthtest{\leftstackrelawd>0pt}}%
{\kern-.5\leftstackrelawd}{}\mathrel{\mathop{#2}\limits^{#1}}}

\newcommand{\Var}{\textup{Var}}

\def\R{{\mathbb R}}

\def\N{{\mathbb N}}

\def\E{{\mathbb E}}
\def\R{{\mathbb R}}

\def\P{{\mathbb P}}

\journal{TBA}

\begin{document}

\begin{frontmatter}

%% Title, authors and addresses

%% use the tnoteref command within \title for footnotes;
%% use the tnotetext command for theassociated footnote;
%% use the fnref command within \author or \address for footnotes;
%% use the fntext command for theassociated footnote;
%% use the corref command within \author for corresponding author footnotes;
%% use the cortext command for theassociated footnote;
%% use the ead command for the email address,
%% and the form \ead[url] for the home page:
%% \title{Title\tnoteref{label1}}
%% \tnotetext[label1]{}
%% \author{Name\corref{cor1}\fnref{label2}}
%% \ead{email address}
%% \ead[url]{home page}
%% \fntext[label2]{}
%% \cortext[cor1]{}
%% \affiliation{organization={},
%%             addressline={},
%%             city={},
%%             postcode={},
%%             state={},
%%             country={}}
%% \fntext[label3]{}

\title{Langevin Bi-fidelity Importance Sampling for Failure Probability Estimation}

%%use optional labels to link authors explicitly to addresses:
\author[1]{Nuojin Cheng}
\author[2]{Alireza Doostan}

\affiliation[1]{organization={Department of Applied Mathematics, University of Colorado Boulder}
}
\affiliation[2]{organization={Smead Aerospace Engineering Sciences Department, University of Colorado Boulder}
}

\begin{abstract}
    Estimating failure probability is a key task in the field of uncertainty quantification. In this domain, importance sampling has proven to be an effective estimation strategy; however, its efficiency heavily depends on the choice of the biasing distribution. An improperly selected biasing distribution can significantly increase estimation error. One approach to address this challenge is to leverage a less expensive, lower-fidelity surrogate. Building on the accessibility to such a model and its derivative on the random uncertain inputs, we introduce an importance sampling-based estimator, termed the Langevin bi-fidelity importance sampling (L-BF-IS), which uses score-function-based sampling algorithms to generate new samples and substantially reduces the mean square error (MSE) of failure probability estimation. The proposed method demonstrates lower estimation error, especially in high-dimensional input spaces and when limited high-fidelity evaluations are available. The L-BF-IS estimator's effectiveness is validated through experiments with two synthetic functions and two real-world applications governed by partial differential equations. These real-world applications involve a composite beam, which is represented using a simplified Euler-Bernoulli equation as a low-fidelity surrogate, and a steady-state stochastic heat equation, for which a pre-trained neural operator serves as the low-fidelity surrogate.
\end{abstract}

% %%Graphical abstract
% \begin{graphicalabstract}
% %\includegraphics{grabs}
% \end{graphicalabstract}

% %%Research highlights
% \begin{highlights}
% \item Research highlight 1
% \item Research highlight 2
% \end{highlights}


\end{frontmatter}

%\linenumbers

%%%%%%%%%%%%%%%%%%%%%%%%%%%%%%%%%%%%%%%%%%%%%%%%%%%%%%%%%%%%%%%%%%%%%%%%%% Introduction
\section{Introduction}
Uncertainty ubiquitously appears in many real-world applications, such as weather forecasting, financial modeling, healthcare decision-making, and engineering design. In computational modeling, uncertainty is often represented by a random vector, defined within a specific probability distribution based on prior knowledge or observation data. One of the key goals of uncertainty quantification (UQ) is to estimate the probability of a device or system failure based on model outputs, also known as the quantity of interest (QoI). There are many methods to solve this problem, including the first-order reliability method (FORM) \cite{hasofer1974exact, der2005first} and its extension to the second order \cite{fiessler1979quadratic}. Other works involve the Monte Carlo sampling method \cite{au2001estimation}. However, the standard Monte Carlo method faces the challenge of slow convergence relative to sample size, especially when the probability of the failure event is small. In practical scenarios, model evaluation demands substantial computational resources, limiting the Monte Carlo method's feasibility. Consequently, there is significant interest in enhancing the convergence of the Monte Carlo method by reducing the number of model evaluations, primarily achieved by reducing the variance of estimators. 

Variance reduction in Monte Carlo estimators can be primarily approached in two ways. The first method involves control variates \cite{asmussen2007stochastic,fairbanks2017low,gorodetsky2020generalized}, which uses correlated random variables to adjust the original estimator based on the covariance between the control and target variables. This adjustment yields a new estimator with reduced variance, provided the appropriately chosen control variable is well-correlated with the target variable. Despite its widespread application in UQ, the efficient control variate method is contingent on the availability of highly correlated control variables with known (or cheap to evaluate) mean and accurate covariance estimation, limiting its applicability. The second method is importance sampling (IS), which is the focus of this work. IS samples the input random variables following a different probability distribution (referred to as biasing distribution) to emphasize the regions that significantly impact the estimation. This approach effectively reduces the estimator variance and focuses on crucial input space areas, proving particularly useful in scenarios involving rare events or tail probability estimations. The critical challenge in IS is constructing a suitable biasing distribution, a task complicated by limited access to model outputs under the UQ setting. 

Several studies have examined the construction of biasing distributions specifically tailored for failure probability estimation. Adaptive importance sampling techniques, such as those in \cite{bucher1988adaptive, kurtz2013cross, geyer2019cross, peherstorfer2018multifidelity}, tune the biasing density within a parameterized family by adaptively finding the optimal density under the cross entropy criteria. Papaioannou et al. \cite{papaioannou2016sequential} discuss the application of sequential importance sampling (SIS) for estimating the probability of failure in structural reliability analysis. Initially developed for exploring posterior distributions and estimating normalizing constants in Bayesian inference, SIS involves a sequential reweighting operation that progressively shifts samples from the prior to the posterior distribution. This work was later adapted using the ensemble Kalman filter \cite{wagner2022ensemble} and consensus sampling \cite{althaus2024consensus}. However, these methods may require extensive forward model computations, limiting their practical applicability. 

To mitigate the computational burdens associated with high-fidelity (HF) models, adopting a ``low-fidelity'' (LF) model proves advantageous. This model, for instance,  derived from the same solver but employing a coarser grid or an approximate surrogate function—either based on fixed basis functions or data-driven—offers reduced accuracy in exchange for significantly lower computational cost. This approach, often named bi-fidelity or multi-fidelity, has been integrated into many of the aforementioned methods. For instance, Li et al. \cite{li2010evaluation} utilized surrogates of the limit state function as low-fidelity models to enhance adaptive importance sampling \cite{li2011efficient}. Similarly, Wagner et al. \cite{wagner2020multilevel} extended sequential importance sampling to multi-level cases where coarse grid solutions serve as low-fidelity models.
% Rather than focusing on the specific design of low-fidelity models, another strand of multi-fidelity modeling offers flexibility by not presupposing any particular format for these models. This adaptability makes the approach suitable for a wider range of scenarios. 
Peherstorfer et al. \cite{peherstorfer2016multifidelity} proposed the multi-fidelity importance sampling method (MF-IS), which constructs the biasing distribution by applying a Gaussian mixture model to inputs whose LF evaluations indicate potential failures, suggesting that inputs failing under LF conditions are likely to fail under HF conditions as well. This strategy preserves the unbiased nature of the importance sampling estimator and does not confine the format of the LF model. Subsequent extensions of this framework \cite{peherstorfer2017combining, kramer2019multifidelity, alsup2023context} have integrated a collection of different estimators and explored the balance between computation and accuracy. However, the aforementioned multi-fidelity methods using polynomial chaos surrogates or based on Gaussian mixture models are known for their rapidly growing complexity with the increase in the dimension of the inputs, denoted as $D$. Moreover, identifying the number of failure clusters for the Gaussian mixture model poses challenges without prior knowledge. In response to the identified challenges, a recent work by Cui et al. \cite{cui2024deep} introduced a deep importance sampling method. This method is notable for its biasing distribution construction with linear complexity $\mathcal{O}(D)$. This was achieved through the push-forward of a reference distribution under a series of order-preserving transformations, each shaped by a squared tensor-train decomposition. While this method offers theoretical and numerical advancements over \cite{peherstorfer2016multifidelity}, challenges related to the training of neural networks and its associated optimization error persist.

In practical applications, low-fidelity models often possess additional properties and information that can be leveraged. For instance, when a low-fidelity model is a simplified model, such as the Euler-Bernoulli equation for beam deflections \cite{cheng2024subsampling, cheng2024bi}, its explicit formulation facilitates simple forward evaluation at minimal cost and provides derivative information. Similarly, when the low-fidelity model is a data-driven surrogate model, the recent development of auto-differentiation-enabled libraries \cite{abadi2016tensorflow, paszke2019pytorch} produces derivatives of the forward surrogate map. These examples highlight the potential of utilizing additional knowledge from low-fidelity models to construct more effective biasing distributions for importance sampling estimators.

In this work, we introduce a new importance sampling estimator, named Langevin Bi-fidelity Importance Sampling (L-BF-IS). By leveraging a new parameterization of the biasing density function and the Metropolis-adjusted Langevin algorithm (MALA) \cite{rossky1978brownian, roberts2002langevin}, this estimator scales favorably in high-dimensioned scenarios ($D\geq 100$). Specifically, the required number of iterations for this algorithm depends on $\mathcal{O}(D^{1/3})$ \cite{freund2022convergence}. The contributions of this work are threefold:
\begin{enumerate}
\item We introduce a new parameterization of the biasing density function leveraging a low-fidelity model; see Equation~\eqref{eq:q-xi}. Two approaches are proposed to tune the only hyper-parameter $\ell$;
\item We analyze the L-BF-IS estimator's statistical properties and estimation performance based on the relation between low-fidelity and high-fidelity models;
\item We empirically demonstrate the effectiveness of the MALA on a multimodal biasing density function and the L-BF-IS performance through synthetic and real-world problems governed by differential equations with high-dimensional random inputs.
\end{enumerate}
The structure of this work is as follows. Section~\ref{sec:method} details the construction and implementation of L-BF-IS, presents a discussion on its error analysis. Section~\ref{sec:experiments} demonstrates the performance of L-BF-IS using three numerical examples. \footnote{The codes are available at \href{https://github.com/CU-UQ/L-BF-IS}{https://github.com/CU-UQ/L-BF-IS}.} Finally, Section~\ref{sec:conclusion} concludes the paper and discusses avenues for future research.

%%%%%%%%%%%%%%%%%%%%%%%%%%%%%%%%%%%%%%%%%%%%%%%%%%%%%%%%%%%%%%%%%%%%%%%%%% Methodology
\section{Langevin Bi-fidelity Importance Sampling Estimator and its Properties}\label{sec:method}
In this section, a detailed motivation, construction, and theoretical analysis of the proposed L-BF-IS estimator is presented. Section~\ref{ssec:background} introduces the concepts of Monte Carlo method and importance sampling. Section~\ref{ssec:estimator} presents the groundwork of L-BF-IS: the designed biasing distribution $q(\bm z)$ and the formulation of L-BF-IS estimator. In Section~\ref{ssec:stat-prop}, the statistical properties of the proposed L-BF-IS estimator, including its unbiasedness, variance, and consistency are discussed. Section~\ref{ssec:lengthscale} includes two approaches to estimate the most important parameter in our estimator, $\ell$. Section~\ref{ssec:sampling} presents the MALA-based technique employed to sample the biasing distribution. A discussion on the influence of the relation between low-fidelity and high-fidelity models on the performance of L-BF-IS estimation is presented in Section~\ref{ssec:bifidelity}. Section~\ref{ssec:error} provides insights on the potential sources of errors in L-BF-IS. 

\subsection{Background}
\label{ssec:background}
We consider an input-output system encompassing an input random vector of dimension $D\in\N$ and an output random variable, named the quantity of interest (QoI), of dimension $d\in\N$. A probability space $(\Omega,\mathcal{F},\mathbb{P})$ is embedded in the input space so that $\Omega\subset \R^D$. The system is represented as a $\mathcal{F}$-measurable function that is equipped with two distinct levels of fidelity: a high-fidelity (HF) QoI function $f^\HF:\Omega\to\R^d$ and a low-fidelity (LF) QoI function $f^\LF:\Omega\to\R^d$, with $\Omega\subset\R^D$. The inputs are random variables $\bm z$ that are assumed to obey an absolutely continuous (with respect to Lebesgue measure) probability distribution, yielding a density function $p(\bm z)$ with associated law $\mathbb{P}_p$. Additionally, for the failure probability, we define Borel-measurable performance functions $g^\HF:\R^d\to\R$ and $g^\LF:\R^d\to\R$. These two functions evaluate the failure result given a QoI and provide a value reflecting the outputs. For simplicity, we define $h^\HF\coloneqq g^\HF\circ f^\HF$ and $h^\LF\coloneqq g^\LF\circ f^\LF$. If $h^\HF, h^\LF(\bm z)<0$, the result represents failures. In the following contexts, we call $h^\HF$ and $h^\LF$ as HF and LF functions, respectively. In the literature \cite{li2010evaluation, li2011efficient}, limit state function that describes $\{\bm z\mid h^\HF(\bm z) = 0\}$ is also discussed. The existence of such a limit state function is based on certain continuity of the function $h^\HF$, which is not assumed in this work. We also define failure regions $\AL\coloneqq (h^\LF)^{-1}((-\infty, 0))$ and $\AH\coloneqq (h^\HF)^{-1}((-\infty, 0))$. Both $\AH$ and $\AL$ belong to $\mathcal{F}$ due to the measurable-function assumption and can be multi-modal. We let $\mathbb{E}_p$ and $\Var_p$ denote the expectation and variance associated with the density $p(\bm z)$, respectively.

Under the multi-fidelity scheme, we aim to evaluate the expected HF failure probability, 
\begin{equation}
    \ExH \coloneqq \P_p[h^\HF(\bm z)<0] = \E_p[\mathbbm{1}_{h^\HF(\bm z)<0}] = \int_\Omega \mathbbm{1}_{h^\HF(\bm z)<0}p(\bm z)d\bm z=\int_{\AH}p(\bm z)d\bm z = \mathbb{P}_p[\AH],
\end{equation}
where $\mathbbm{1}$ is the indicator function. The Monte Carlo estimator, with $N$ samples, is 
\begin{equation}
\label{eq:mc-def}
\widehat{P}^\textup{MC}_N \coloneqq \frac{1}{N}\sum_{i=1}^N \mathbbm{1}_{h^\HF(\bm z_i)<0},\quad \{\bm z_i\}_{i=1}^N\overset{\mathrm{iid}}{\sim}p(\bm z).
\end{equation}
In this work, we use the hat notation to denote estimators. The mean square error (MSE) of Monte Carlo estimation $\E_p[(\widehat{P}^\textup{MC}_N - \ExH)^2]$ is $\Var_p[\mathbbm{1}_{h^\HF(\bm z)<0}]/N$. In applications like failure probability estimation, when the failure probability is small, the aforementioned variance becomes large, which requires more HF evaluations to reduce the MSE. Importance sampling (IS) \cite{penfold1967monte} is one of the methods that effectively reduces the estimator variance by re-weighting the samples with a carefully chosen alternative density function $q(\bm z)$, named biasing density function. By building $q(\bm z)$ to replace $p(\bm z)$, the IS estimator is then defined as
\begin{align}
\label{eq:is-def}
    \widehat{P}^\textup{IS}_N \coloneqq \frac{1}{N}\sum_{i=1}^N \mathbbm{1}_{h^\HF(\tilde{\bm z}_i)<0}\frac{p(\tilde{\bm z}_i)}{q(\tilde{\bm z}_i)},\quad \{\tilde{\bm z}_i\}_{i=1}^N\overset{\mathrm{iid}}{\sim} q(\bm z).
\end{align}
Note that the IS estimator in Equation~\eqref{eq:is-def} is unbiased, i.e., $\mathbb{E}_q[\widehat{P}^\textup{IS}_N]=\ExH$.

\subsection{Biasing Distribution and L-BF-IS Estimator}
\label{ssec:estimator}
It is known that the optimal biasing density for failure probability estimation is (see \cite{penfold1967monte})
\begin{equation}
\label{eq:opt-q}
q^\ast(\bm z) \coloneqq \frac{1}{\ExH}\mathbbm{1}_{h^\HF(\bm z) < 0}p(\bm z).
\end{equation}
However, we cannot simply use the LF indicator function $\mathbbm{1}_{h^\LF(\cdot)<0}$ to replace its counterpart $\mathbbm{1}_{h^\HF(\cdot)<0}$ due to singularity issue on the IS weight $p(\bm z)/q(\bm z)$.  Instead, we aim to design a ``soft version'' for the conceptually optimal biasing density while providing it with flexibility to adjust for unmatching support between $\mathbbm{1}_{h^\HF(\cdot)<0}$ and $\mathbbm{1}_{h^\LF(\cdot)<0}$.

Similar to the smoothing strategy in \cite{papaioannou2019improved, uribe2021cross}, we propose the biasing distribution
\begin{align}
\label{eq:q-xi}
    q(\bm z) \coloneqq \frac{1}{\mathcal{Z}(\ell)}\exp\left(-\ell\tanh\circ h^\LF(\bm z) \right)p(\bm z),
\end{align}
where $\ell$ is a length scale and $\mathcal{Z}(\ell)$, a function of $\ell$, is the normalisation constant. The value of $\mathcal{Z}(\ell)$ is given by
\begin{equation}
\label{eq:z-def}
\mathcal{Z}(\ell) = \int_{\Omega} \exp\left(-\ell\tanh\circ h^\LF(\bm z) \right)p(\bm z) d\bm z = \mathbb{E}_p\left[\exp\left(-\ell\tanh\circ h^\LF(\bm z) \right)\right].
\end{equation}
Note that $q(\bm z)$ is strictly positive when the input is in the support of $p(\bm z)$, which guarantees that $p(\bm z)$ is absolutely continuous with respect to $q(\bm z)$ and the weight $p(\bm z)/q(\bm z)$ is well-defined. Based on the initial density $p(\bm z)$, the formulation of $q(\bm{z})$ in Equation~\eqref{eq:q-xi} prioritizes higher probability weights for samples $\bm{z}$ whose LF counterparts indicate a failure outcome. This approach is based on an assumed connection between the HF function $h^\HF$ and its LF counterpart $h^\LF$, which will be discussed in more details in Section~\ref{ssec:bifidelity}. Figure~\ref{fig:q-example} illustrates this concept with a two-dimensional ($D=2$) example, demonstrating the application of our proposed method. Similar to the strategy employed by \cite{li2011efficient}, the $\tanh$ function facilitates a ``buffer'' region within the importance sampling framework. However, unlike the method above, our approach does not aim to directly approximate limit state functions due to its complexities in high-dimensional space. 

\begin{figure}[ht]
	\centering
	\includegraphics[width = 1.0\textwidth]{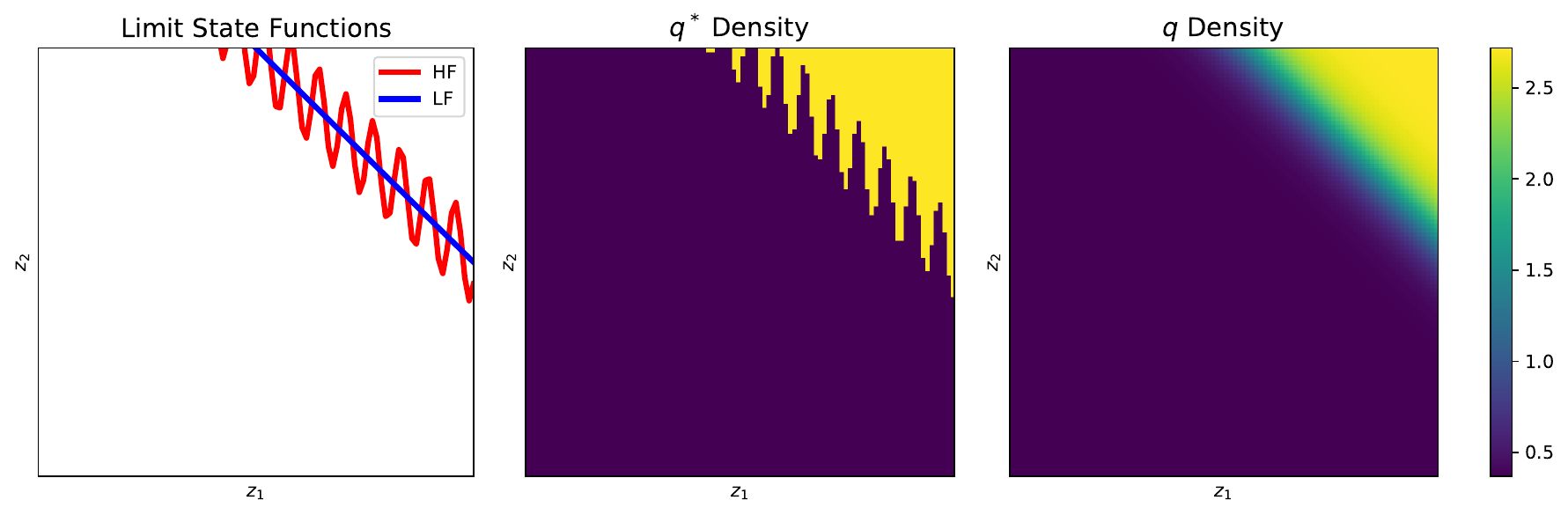}
	\caption{Illustration of the concept of limit state functions and biasing densities in the inputs $\bm{z}$. The left figure displays the limit state functions that separate the failure region from the safe region, highlighting the HF limit function in red and the LF surrogate in blue. The middle figure shows the optimal biasing density as derived from Equation~\eqref{eq:opt-q}. The right figure displays the proposed biasing density, as defined in Equation~\eqref{eq:q-xi}, which utilizing the LF function.}
	\label{fig:q-example}
\end{figure}

Given $q(\bm z)$ in Equation~\eqref{eq:q-xi} and $p(\bm z)$, the importance sampling weight function is
\begin{align}
\label{eq:weight-def}
    \frac{p(\bm z)}{q(\bm z)} = \mathcal{Z}(\ell)\exp\left(\ell\tanh\circ h^\LF(\bm z) \right).
\end{align}
We approximate $\mathcal{Z}(\ell)$ using Monte Carlo estimation 
\begin{equation}
\label{eq:z-est}
\widehat{\mathcal{Z}}_M(\ell) = \frac{1}{M}\sum_{m=1}^M \exp\left(-\ell\tanh\circ h^\LF(\bm z_m) \right),
\end{equation}
where $\{\bm z_m\}_{m=1}^M\overset{\mathrm{iid}}{\sim}p(\bm z)$. Since the estimation of $\mathcal{Z}(\ell)$ only involves evaluating the inexpensive LF function, $M$ can be sufficiently large so that $\mathcal{Z}(\ell)$ can be estimated with high accuracy. 

According to the definition of importance sampling estimator in Equation~\eqref{eq:is-def}, and given our $q(\bm z)$, we define the L-BF-IS estimator as follows
\begin{align}
\label{eq:bf-is-estimator}
    \widehat{P}^\BF_{M,N} = \left(\frac{1}{M}\sum_{m=1}^M \exp\left(-\ell\tanh\circ h^\LF(\bm z_m) \right)\right)\left(\frac{1}{N}\sum_{i=1}^N \mathbbm{1}_{h^\HF(\tilde{\bm z}_i)<0}\exp\left(\ell\tanh\circ h^\LF(\tilde{\bm z}_i) \right)\right),\hspace{1.0em}
\end{align}
where $\{\bm z_m\}_{m=1}^M\overset{\mathrm{iid}}{\sim}p(\bm z)$ and $\{\tilde{\bm z}_i\}_{i=1}^N\overset{\mathrm{iid}}{\sim} q(\bm z)$.

%%%%%%%%%
\subsection{Statistical Properties of L-BF-IS Estimator}
\label{ssec:stat-prop}
Analyzing bias, variance, and consistency, is crucial for evaluating the performance of an estimator. Firstly, due to the independence between samples from $p(\bm z)$ and $q(\bm z)$, the L-BF-IS estimator in Equation~\eqref{eq:bf-is-estimator} is unbiased, i.e.,
\begin{equation}
\label{eq:bf-is-mean}
\begin{aligned}
\mathbb{E}_{p\otimes q}\left[\widehat{P}^\BF_{M,N}\right] &= \mathbb{E}_p\left[\exp\left(-\ell\tanh\circ h^\LF(\bm z) \right)\right]\mathbb{E}_q\left[\mathbbm{1}_{h^\HF(\bm z)<0}\exp\left(\ell\tanh\circ h^\LF(\bm z) \right)\right]\\
&= \mathcal{Z}(\ell)\mathbb{E}_q\left[\mathbbm{1}_{h^\HF(\bm z)<0}\exp\left(\ell\tanh\circ h^\LF(\bm z) \right)\right] = \ExH.
\end{aligned}
\end{equation}
Here, $p\otimes q$ represents the Cartesian product of the two densities, indicating their independence and the last equality is from the unbiasedness of the important sampling estimator. Secondly, following the relation
\begin{equation}
\label{eq:ind-var}
\Var[XY] = \Var[X]\Var[Y] + \E^2[X]\Var[Y] + \Var[X]\E^2[Y],
\end{equation}
for two independent variables $X$ and $Y$ the variance of L-BF-IS estimator is given by 
\begin{equation}
\label{eq:bf-is-var}
\begin{aligned}
\Var_{p\otimes q}[\widehat{P}^\BF_{M,N}] &= \frac{1}{MN}\Var_p\left[\exp\left(-\ell\tanh\circ h^\LF(\bm z) \right)\right]\Var_q\left[\mathbbm{1}_{h^\HF(\bm z)<0}\exp\left(\ell\tanh\circ h^\LF(\bm z) \right)\right]\\ 
&\quad+ \frac{\mathcal{Z}^2(\ell)}{N}\Var_q\left[\mathbbm{1}_{h^\HF(\bm z)<0}\exp\left(\ell\tanh\circ h^\LF(\bm z) \right)\right]\\ 
&\quad+ \frac{1}{M}\Var_p\left[\exp\left(-\ell\tanh\circ h^\LF(\bm z) \right)\right](\ExH)^2.
\end{aligned}
\end{equation}
With the results from Equation~\eqref{eq:bf-is-mean} and Equation~\eqref{eq:bf-is-var}, the consistency of the L-BF-IS can be shown by applying the Chebyshev's inequality, 
\begin{equation}
\mathbb{P}_{p\otimes q}\left(\lvert\widehat{P}^\BF_{M, N} - \ExH\rvert 
\geq \epsilon\right) \leq \frac{\Var_{p\otimes q}\left[\widehat{P}^\BF_{M, N}\right]}{\epsilon^2},\quad \forall \epsilon > 0.
\end{equation}
Notice that the variance $\Var_{p\otimes q}\left[\widehat{P}^\BF_{M, N}\right]$ decays when both $M$ and $N$ increase.
Additionally, if we assume the value of $M$ is sufficiently large so that $1/M$ is small enough to be ignored, the variance in Equation~\eqref{eq:bf-is-var} can be approximated as 
\begin{equation}
\label{eq:bf-is-var-apx}
\Var_{p\otimes q}[\widehat{P}^\BF_{M,N}] \approx \frac{\mathcal{Z}^2(\ell)}{N}\Var_q\left[\mathbbm{1}_{h^\HF(\bm z)<0}\exp\left(\ell\tanh\circ h^\LF(\bm z) \right)\right].
\end{equation}

%%%%%%%%%
\subsection{Selection of Lengthscale $\ell$}
\label{ssec:lengthscale}
The value of the parameter $\ell$ in Equation~\eqref{eq:q-xi} plays a key role in determining the performance of the L-BF-IS estimator. Given that the estimator is unbiased as shown in Equation~\eqref{eq:bf-is-mean} and the values of $M$ and $N$ are held fixed, the goal is to find an optimal value of $\ell$ so that the variance of the L-BF-IS estimator is minimized. Leveraging the variance approximation presented in Equation~\eqref{eq:bf-is-var-apx} and acknowledging the dependency of $q(\bm z)$ on $\ell$, we re-formulate the approximated variance as
\begin{equation}
\label{eq:bf-is-var-apx-p}
\begin{aligned}
\Var_{p\otimes q}\left[\widehat{P}^\BF_{M,N}\right]
%&\approx \frac{\mathcal{Z}^2(\ell)}{N}\Var_q\left[\mathbbm{1}_{h^\HF(\bm z)<0}\exp\left(\ell\tanh\circ h^\LF(\bm z) \right)\right]\\
&\approx \frac{\mathcal{Z}(\ell)}{N} \E_p\left[\mathbbm{1}_{h^\HF(\bm z)<0}\exp\left(\ell\tanh\circ h^\LF(\bm z) \right)\right] - \frac{(\ExH)^2}{N}.
\end{aligned}
\end{equation}
For the interest of brevity, more detils on the derivation of Equation~\eqref{eq:bf-is-var-apx-p} are presented in \ref{apdx:variance}. Focusing solely on the relationship between the variance in Equation~\eqref{eq:bf-is-var-apx-p} and $\ell$, 
\begin{equation}
\label{eq:bf-is-var-apx-p-2}
\Var_{p\otimes q}\left[\widehat{P}^\BF_{M,N}\right] \approx \underbrace{\frac{\mathcal{Z}(\ell)}{N}}_{\ell\downarrow} \underbrace{\E_p\left[\mathbbm{1}_{h^\HF(\bm z)<0}\exp\left(\ell\tanh\circ h^\LF(\bm z) \right)\right]}_{\ell\uparrow} + \mathcal{O}(1).
\end{equation}
Upon examining Equation~\eqref{eq:bf-is-var-apx-p-2} closely, it is clear that the value of $\mathcal{Z}(\ell)$, as defined in Equation~\eqref{eq:z-def}, decreases monotonically with $\ell$ while the expectation component exhibits a monotonic increase with the value of $\ell$. This dichotomy highlights a trade-off between larger and smaller $\ell$ values, underscoring the importance of designing an algorithm to optimally determine $\ell$.

Since estimating the expectation term requires evaluating the HF function $h^\HF$, two practical approaches are next introduced to choose an optimal value for $\ell$.
%%%%%%%%%%%%%
\subsubsection{Approach One: Using Pilot HF Evaluations}
\label{sssec:ell-one}
In the first approach, one consider a small sample approximation of the variance in \eqref{eq:bf-is-var-apx-p-2}
\begin{equation}
\label{eq:var-est1}
\begin{aligned}
\widehat{V}_L(\ell) &= \frac{\widehat{\mathcal{Z}}_M(\ell)}{NL} \sum_{j=1}^L\mathbbm{1}_{h^\HF(\bm z_j)<0}\exp\left(\ell\tanh\circ h^\LF(\bm z_j) \right), \quad \{\bm z_j\}_{j=1}^L\sim p(\bm z)
\end{aligned}
\end{equation}
using $L\ll M$ HF function $h^\HF(\cdot)$ evaluations. We then choose the optimal $\ell^*$ such that 
\begin{equation}
\ell^* = \argmin_{\ell} \widehat{V}_L(\ell),
\end{equation}
which, as a 1D optimization problem, can be solved using a simple grid search or a first/second order method. However, when the failure probability is small, e.g. $\ExH \leq \mathcal{O}(1/L)$, a risk of this approach is that $\mathbbm{1}_{h^\HF(\bm z_j)<0}$ can be $0$ for all $\bm z_j$, thus making it invalid. Indeed, since the HF function is evaluated only $L$ times, the probability that no failure case is sampled is $(1 - \ExH)^L$ and can be non-trivial.

%%%%%%%%
\subsubsection{Approach Two: Only Using LF Evaluations}
\label{sssec:ell-two}
An alternative approach is to replace $\mathbbm{1}_{h^\HF(\bm z_j)<0}$ with $\mathbbm{1}_{h^\LF(\bm z_j)<0}$, which produces the variance estimator
\begin{equation}
\label{eq:var-est2}
\begin{aligned}
\widehat{V}'_M(\ell) &= \frac{\widehat{\mathcal{Z}}_M(\ell)}{NM} \sum_{m=1}^M\mathbbm{1}_{h^\LF(\bm z_m)<0}\exp\left(\ell\tanh\circ h^\LF(\bm z_j) \right),
\end{aligned}
\end{equation}
with samples $\{\bm z_m\}_{m=1}^M\overset{\mathrm{iid}}{\sim}p(\bm z)$. We choose the optimal $\ell^*$ as 
\begin{equation}
\ell^* = \argmin_\ell \widehat{V}'_M(\ell).
\end{equation}
This approach provides a less accurate estimation for the variance in exchange for avoiding directly evaluating the HF function. We suggest applying this approach when the value of $(1-\ExH)^L$ is large, where $\ExH$ can be replaced by some prior knowledge of the failure probability and, $L\ll M$, is an affordable number of HF function evaluations.

% Our experiments in Section~\ref{sec:experiments} implement both approaches and their results show that approach two yields better performance in most cases.
%
%
\subsection{Sampling the Biasing Distributions}
\label{ssec:sampling}
The formulation of the biasing density $q(\bm{z})$ in Equation~\eqref{eq:q-xi}, as well as the availability of the LF function derivative $\nabla_{\bm{z}} h^\LF(\bm{z})$ facilitates the evaluation of the score function $\nabla_{\bm{z}} \log q(\bm{z})$. This capability significantly enhances the selection of sampling methods that utilize the score function, including but not limited to Langevin Monte Carlo, Hamiltonian Monte Carlo, and Stein Variational Gradient Descent \cite{liu2016stein}.

Among the various options, we opt for the Metropolis-adjusted Langevin algorithm (MALA), a variant of Langevin Monte Carlo, to generate samples from the biasing distribution. This choice is made because of its simplicity and widely-used implementation. However, it is important to note that any score-based sampling method is compatible with the importance sampling framework proposed in this work. MALA effectively integrates the discretization of Langevin dynamics with the Metropolis-Hastings algorithm \cite{roberts2002langevin}, offering a robust framework for sampling. 

Assuming the score function $\nabla_{\bm z}\log p(\bm z)$ exists and is bounded, and the LF function $h^\LF$ is differentiable and Lipschitz, the biasing density $q(\bm z)$ can be written as 
\begin{align}
\label{eq:U}
    q(\bm z) = \frac{1}{\mathcal{Z}(\ell)}\exp(-U(\bm z)),
\end{align}
where the potential function $U(\bm z)$ is given by 
\begin{align}
\label{eq:U-def}
    U(\bm z) \coloneqq \ell\tanh\circ h^\LF(\bm z)-\log p(\bm z).
\end{align}
According to \cite{roberts2002langevin}, the density $q(\bm z)$ is the unique invariant distribution of the Langevin stochastic differential equation (SDE)
\begin{align}
\label{eq:langevin-sde}
    d\bm z = -\nabla U(\bm z) + \sqrt{2}d\bm{W}_t,
\end{align}
where $\bm{W}_t$ is the Brownian motion. 
Therefore, by simulating the SDE in Equation~\eqref{eq:langevin-sde} via Euler-Maruyama method,
\begin{align}
\label{eq:sde-disc}
    \bm z^{(t+\tau)} = \bm z^{(t)} - \tau\nabla U(\bm z^{(t)}) + \sqrt{2}(\bm{W}_{t+\tau}-\bm{W}_t),
\end{align}
where $\bm z^{(t)}$ represents the discretized $\bm z$ and $\tau$ is the step size. 
The property of Brownian motion, $\bm{W}_{t+\tau}-\bm{W}_t\sim\mathcal{N}(\bm{0},\tau \bm{I}_D)$ with the identity matrix $\bm{I}_D\in\R^{D\times D}$, allows to re-write Equation~\eqref{eq:sde-disc} as 
\begin{align}
\label{eq:langevin-alg}
    \bm z^{(t+1)} &= \bm z^{(t)} - \tau\nabla U(\bm z^{(t)}) + \sqrt{2\tau}\bm\epsilon,\quad \bm\epsilon\sim\mathcal{N}(\bm{0},\bm{I}_D).
\end{align}
Following Equations~\eqref{eq:U} and \eqref{eq:U-def}, $\nabla_{\bm z} U(\bm z)$ is given by
\begin{equation}
\begin{aligned}
\label{eq:score}
    \nabla_{\bm z} U(\bm z) 
    &= -\nabla_{\bm z} \log q(\bm z) 
    = \ell\nabla_{\bm z}\tanh\circ h^\LF(\bm z) - \nabla_{\bm z} \log p(\bm z).
%    &= \ell\sech^2\left(h^\LF(\bm z)\right)\nabla_{\bm z} h^\LF(\bm z) - \nabla_{\bm z} \log p(\bm z).
\end{aligned}
\end{equation}
Besides sampling $\bm z^{(t)}$ iteratively, MALA implements a Metropolis-Hastings accept-reject mechanism to reject proposals in low-density regions \cite{roberts1996exponential}. The rejection of the new proposed sample $\bm z^{(t+1)}$ is triggered if 
\begin{align}
\label{eq:metropolis}
u \geq \exp\left(U(\bm z^{(t)}) + \pi(\bm z^{(t)}, \bm z^{(t+1)}) - U(\bm z^{(t+1)}) - \pi(\bm z^{(t+1)}, \bm z^{(t)})\right),
\end{align}
where $u$ is a random variable sampled from uniform distribution $U[0,1]$ and $\pi$ is a function defined as 
\begin{align}
\pi(\bm z_1, \bm z_2) \coloneqq -\frac{1}{4\tau}\lVert \bm z_1 - \bm z_2 - \tau\nabla U(\bm z_2)\rVert_2^2.
\end{align}
Numerically, we discard the first $B$ samples of the Markov chain, referred to as burn-in samples, with $B$ varying depending on the problem scale. The sampling algorithm is detailed in Algorithm~\ref{alg:langevin-sampling}. Notice that, once $\ell$ is set, Algorithm~\ref{alg:langevin-sampling} requires only $\mathcal{O}(T+B)$ LF evaluations. The construction of the L-BF-IS estimator is concluded in Algorithm~\ref{alg:l-bf-is}.

\begin{algorithm}[ht]
		\DontPrintSemicolon
		\caption{Langevin Algorithm for Sampling from Biasing Distribution\label{alg:langevin-sampling} \\ {\color{blue}$\mathcal{O}(T+B)$ LF evaluations}}
		\KwIn{Length scale $\ell$, burn-in number $B$, LF function $h^\LF$, its gradient $\nabla h^\LF$, step size $\tau$, iteration number $T$, and initial state $\bm z^{(0)}$ (Optional)}
		\KwOut{A collection of samples $\{\tilde{\bm z}_i\}_{i=1}^{N}\overset{}{\sim}q(\bm z)$}
		\begin{algorithmic}[1]
			\STATE Sample initial state $\bm z^{(0)}\overset{\mathrm{iid}}{\sim}p(\bm x)$ if $\bm z^{(0)}$ is not given
                \FOR{$t = 1:T+B$} 
			\STATE update $\bm z^{(t)}$ following Equation~\eqref{eq:langevin-alg} and Equation~\eqref{eq:score}
                \STATE reject the step if Equation~\eqref{eq:metropolis} satisfied.
                \ENDFOR
            \STATE $\{\tilde{\bm z}_t\}_{t=1}^{T}\gets\{\bm z^{(t)}\}_{t=B+1}^{T+B}$
		\end{algorithmic}
\end{algorithm}

\begin{remark}
    When implementing Algorithm~\ref{alg:langevin-sampling} on a bounded domain $\Omega$, we introduce a penalty value $ q(\bm{z}) \gg 0 $ for all $ \bm{z} \notin \Omega $ to discourage the chain from moving outside the domain.
\end{remark}

\begin{algorithm}[ht]
		\DontPrintSemicolon
		\caption{L-BF-IS Method\label{alg:l-bf-is}
  \\ {\color{blue}$\mathcal{O}(M+T+B)$ LF evaluations}\quad{\color{red}$\mathcal{O}(N+L)$ HF evaluations}}
		\KwIn{LF sample size $M$, HF sample size $N$, LF function $h^\LF$, HF function $h^\HF$, and additional HF sample size $L$ (optional)}
		\KwOut{A value of L-BF-IS estimator $\widehat{P}^\BF_{M,N}$}
		\begin{algorithmic}[1]
			\STATE Determine Langevin dynamics step size $\tau$, burn-in number $B$, and iteration number $T$ based on available computational resource $(T>N)$
                \IF{$L$ is provided}
                \STATE Determine $\ell$ that minimizes the variance estimator in Equation~\eqref{eq:var-est1};
                \ELSE
                \STATE Determine $\ell$ that minimizes the variance estimator in Equation~\eqref{eq:var-est2};
                \ENDIF
                \STATE Build estimator $\widehat{\mathcal{Z}}_M(\ell)$ using $\{\bm z_m\}_{m=1}^M\overset{\mathrm{iid}}{\sim}p(\bm z)$ following Equation~\eqref{eq:z-est};
                \STATE $\{\tilde{\bm z}_t\}_{t=1}^{T}\gets\textup{Langevin algorithm}(\ell,B,h^\LF, \nabla h^\LF,\tau,T)$ in Algorithm~\ref{alg:langevin-sampling};
                \STATE Uniformly select subset $\{\tilde{\bm z}_i\}_{i=1}^{N}\subseteq \{\tilde{\bm z}_t\}_{t=1}^{T}$
                \STATE Evaluate $\widehat{P}^\BF_{M, N}$ as in Equation~\eqref{eq:bf-is-estimator} using $\{\tilde{\bm z}_i\}_{i=1}^{N}$ and $\widehat{\mathcal{Z}}_M(\ell)$.
		\end{algorithmic}
\end{algorithm}
%

%%%%%%%%%%%%%%%%%%%%%%%%%%%%%%%%%%%%%%%%%%%%%%%%%%%
\subsection{Further Discussion on Bi-fidelity Modeling}
\label{ssec:bifidelity}
A crucial aspect of any bi-fidelity modeling is understanding how the similarity between LF and HF models affects the performance of the proposed bi-fidelity algorithm, while we investigate from two perspectives: the variance of the L-BF-IS estimator and the Kullback-Leibler (KL) divergence between the optimal and the proposed biasing distributions.

Recall that in Section~\ref{ssec:background} we define subsets $\AH \subset \Omega$ and $\AL \subset \Omega$ such that $\bm{z} \in \AH$ if and only if $h^{\HF}(\bm{z}) < 0$, and $\bm{z} \in \AL$ if and only if $h^{\LF}(\bm{z}) < 0$. Note that under this definition, $\ExH = \P_p[\AH]$. The analysis in this section assumes $\ell$ is already fixed. Based on the approximated variance in Equation~\eqref{eq:bf-is-var-apx-p}, we decompose the expectation term into two parts:
\begin{equation}
\label{eqn:var-bound}
\begin{aligned}
&\E_p\left[\mathbbm{1}_{h^\HF(\bm z)<0}\exp\left(\ell\tanh\circ h^\LF(\bm z) \right)\right]
= \int_{\AH} \exp\left(\ell\tanh\circ h^\LF(\bm z) \right) p(\bm z)d\bm z\\
&= \int_{\AH\cap \AL} \exp\left(\ell\tanh\circ h^\LF(\bm z) \right) p(\bm z)d\bm z + \int_{\AH\cap \AL^C} \exp\left(\ell\tanh\circ h^\LF(\bm z) \right) p(\bm z)d\bm z,
\end{aligned}
\end{equation}
where $\AL^C \coloneqq \Omega \setminus \AL$ is the complement. For the first term in Equation~\eqref{eqn:var-bound}, since $\bm z\in \AL$, we have $h^\LF(\bm z)<0$ and thus $\tanh\circ h^\LF(\bm z) < 0$, making this term upper bounded by $\P_p[\AH\cap \AL]$, which is equivalent to $\ExH - \P_p[\AH\cap \AL^C]$. The second term, since $\tanh\circ h^\LF(\bm z) < 1$ for all $\bm z$, is thereby bounded above by $e^\ell \P_p[\AH\cap \AL^C]$. 

Applying a similar methodology, we also bound $\mathcal{Z}(\ell) < 1+(e^\ell-1)\P_p[\AL]$; see \ref{lm:cont-bound} for detailed proofs. Thus, assuming $M$ is sufficiently large, the variance of the L-BF-IS estimator in Equation~\eqref{eq:bf-is-var-apx} is upper bounded as:
\begin{equation}
\Var_{p\otimes q}[\widehat{P}^\BF_{M,N}] \lesssim \frac{1+(e^\ell-1)\P_p[\AL]}{N}(\ExH + (e^\ell-1) \P_p[\AH\cap \AL^C]) - \frac{(\ExH)^2}{N}.
\end{equation}
The terms $\P_p[\AL]$ and $e^\ell \P_p[\AH\cap \AL^C]$ represent penalties arising from mismatches between the HF and LF models. Should the LF model perfectly align with the HF model, these terms vanish; see Figure~\ref{fig:AHAL-4cases}. This bound elucidates that the performance of L-BF-IS is contingent on $\P_p[\AH\cap \AL^C]$, and further analysis of the KL divergence will verify this observation.

In addition to the variance analysis, we examine the KL divergence between the proposed biasing distribution in Equation~\eqref{eq:q-xi} and the optimal distribution in Equation~\eqref{eq:opt-q}, given by
\begin{equation}
\label{eq:kl-div}
\KL(q^\ast\|q) = \E_{q^\ast}\left[\log\frac{\mathcal{Z}(\ell)\mathbbm{1}_{h^\HF(\bm z) < 0}}{\ExH\exp\left(-\ell\tanh\circ h^\LF(\bm z)\right)}\right],
\end{equation}
or its simplification
\begin{equation}
\KL(q^\ast\|q) = \log\frac{\mathcal{Z}(\ell)}{\ExH} + \ell\int_{\AH\cap \AL}\tanh\circ h^\LF(\bm z)p(\bm z)d\bm z + \ell\int_{\AH\cap \AL^C}\tanh\circ h^\LF(\bm z)p(\bm z)d\bm z.
\end{equation}
Here, the integrals represent contributions from the regions where high-fidelity and low-fidelity models coincide and where they do not, respectively. Consequently, the KL divergence can be bounded by
\begin{equation}
\label{eq:kl-bound}
\KL(q^\ast\|q) < \log\frac{1+(e^\ell-1)\P_p[\AL]}{\ExH} + \ell\P_p[\AH\cap \AL^C].
\end{equation}
The expression in Equation~\eqref{eq:kl-bound} indicates that the optimality of the proposed biasing distribution depends significantly on $\P_p[\AH\cap \AL^C]$. The proofs supporting these claims are provided in \ref{apdx:simp-kl}. Note that since the optimal biasing distribution $q^*$ is fixed, the KL divergence is equivalent to the cross entropy criteria presented in \cite{kurtz2013cross, geyer2019cross}. 

\begin{figure}[ht!]
    \centering
    \begin{minipage}{0.23\textwidth}
        \centering
        \includegraphics[width=\linewidth, trim={4cm 4cm 4cm 4cm}, clip]{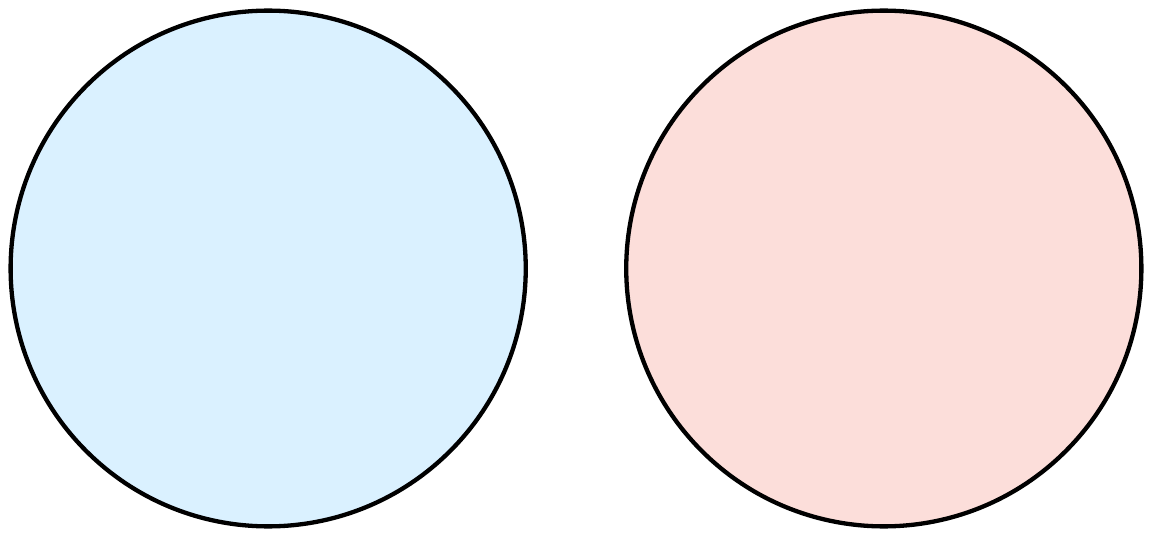}
        \subcaption{}
    \end{minipage}%
    \hfill
    \begin{minipage}{0.23\textwidth}
        \centering
        \includegraphics[width=\linewidth, trim={6cm 5cm 4cm 4cm}, clip]{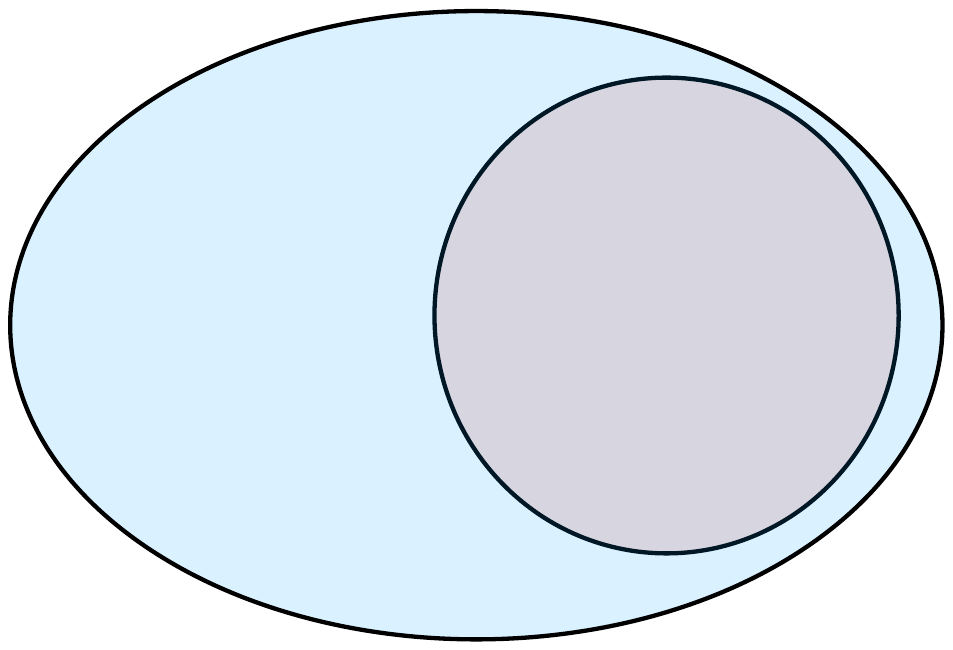}
        \subcaption{}
    \end{minipage}
    \hfill
    \begin{minipage}{0.23\textwidth}
        \centering
        \includegraphics[width=\linewidth, trim={6cm 6cm 6cm 4cm}, clip]{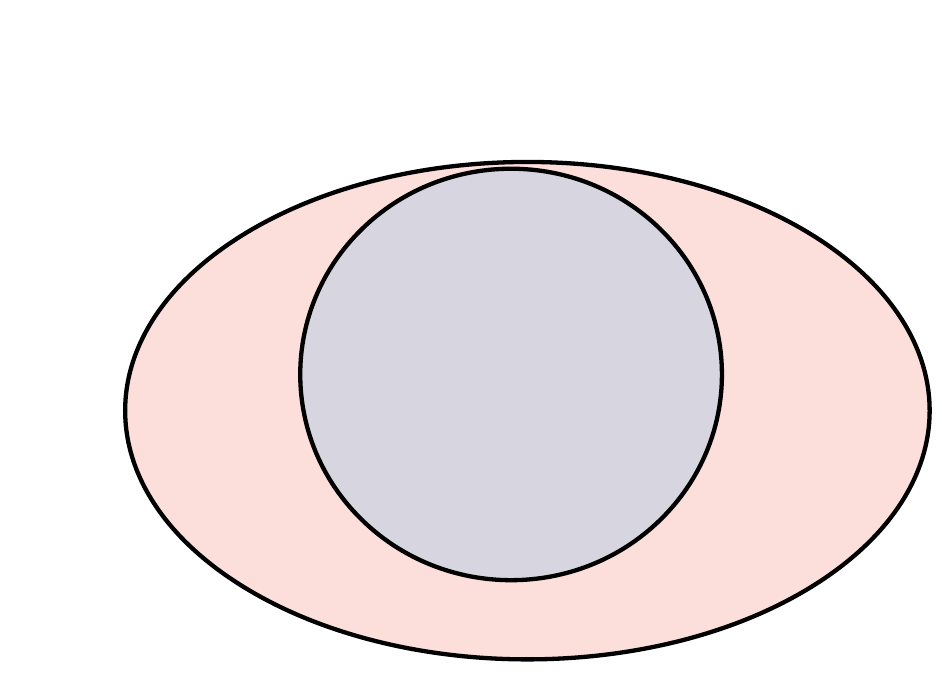}
        \subcaption{}
    \end{minipage}%
    \hfill
    \begin{minipage}{0.23\textwidth}
        \centering
        \includegraphics[width=\linewidth, trim={5.2cm 5cm 4cm 4cm}, clip]{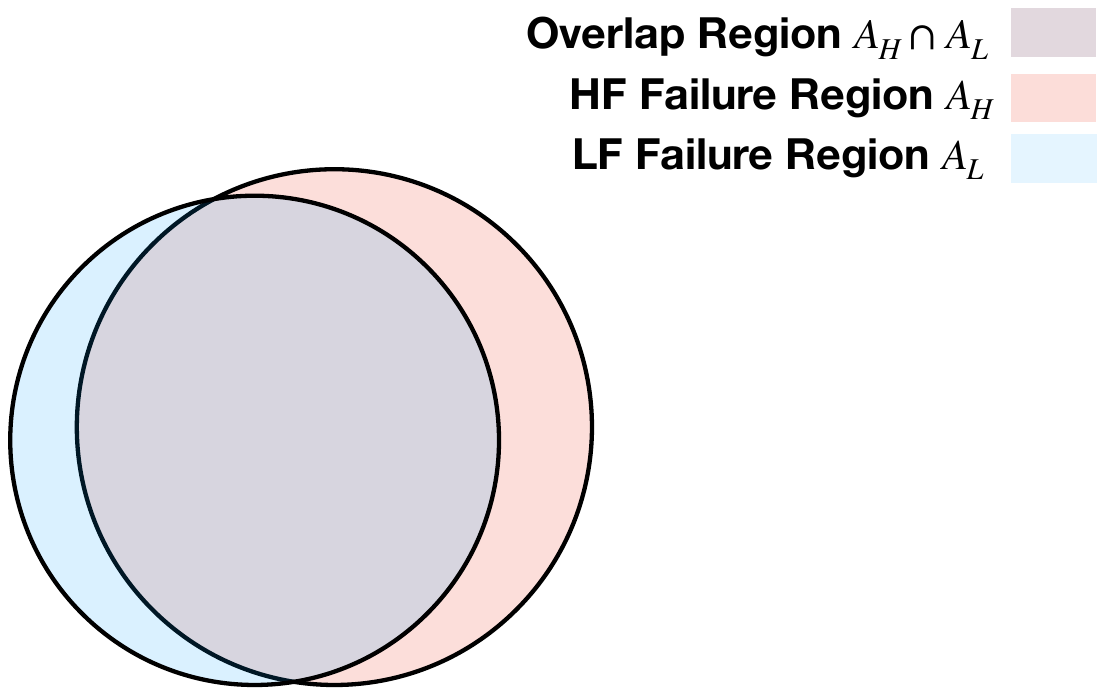}
        \subcaption{}
    \end{minipage}
    \caption{Illustration of the trade-off between $\P_p[\AL]$ and $\P_p[\AH \cap \AL^C]$ when $D=2$. Case 1 (a) represents the worst scenario, where there is no overlap between $\AH$ and $\AL$. In Case 2 (b), we observe an extreme case where $\P_p[\AH \cap \AL^C]$ is zero, but $\P_p[\AL]$ becomes excessively large. Case 3 (c) presents a scenario where $\P_p[\AL]$ is small, but $\P_p[\AH \cap \AL^C]$ is significantly large. Lastly, Case 4 (d) shows a favorable scenario resulting in a small values for both $\P_p[\AL]$ and $\P_p[\AH \cap \AL^C]$. }
    \label{fig:AHAL-4cases}
\end{figure}

While $\P_p[\AH\cap \AL^C]$ is a key in describing the alignment between the HF and LF models, computing it requires many HF model evaluations, which is inpractical. One possible way to address this problem is to use a small number of pilot HF samples to evaluate the KL divergence in Equaiton~\eqref{eq:kl-div}. A systematic framework of the alignment between the LF and HF models is out of the scope of this work but can be the focus of the future works.

%%%%%%%%%%%%%%%%%%%%%%%%%%%%%%%%%%%%%%%%%%%%%%%%%%%
\subsection{Error Analysis}
\label{ssec:error}

Two principal types of errors are identified as contributing to an increase in the MSE: bias-inducing error and variance-inducing error. This section delves into both those error types.

The bias-inducing error arises from inaccuracies in MALA, as outlined in Section~\ref{ssec:sampling}. A series of studies have investigated the convergence behavior of Langevin Monte Carlo, especially under the convexity assumption of the potential function $U(\bm z)$ in Equation~\eqref{eq:U-def}. These studies have shown that Langevin algorithm's output tends to converge to the target distribution $q(\bm z)$ across several metrics, including total variation \cite{dalalyan2017theoretical, durmus2017nonasymptotic}, Wasserstein-2 distance \cite{durmus2019analysis}, and KL divergence \cite{cheng2018convergence}. 
% Originally, it was understood that the necessary iteration number of Langevin Monte Carlo to be weakly related to input dimension, i.i. $\mathcal{O}(D^{1/3})$. The recent work in \cite{freund2022convergence} suggestes that the iterations of the Langevin algorithm required for $\epsilon$-convergence do not depend on the dimension $D$. 
However, the convexity of $U(\bm z)$ may not always hold, particularly for target densities $q(\bm z)$ with multimodal features. The inaccurate sampling of $q(\bm z)$ lead to biases in L-BF-IS estimations. A mitigation strategy involves launching multiple Langevin dynamics chains from different initial states. % However, its convergence is not guaranteed. 

The variance-inducing error originates from two sources. The first is the discrepancy between LF and HF functions. According to the analysis in Section~\ref{ssec:bifidelity}, this discrepancy is quantified by the probabilities $\P_p[\AH\cap \AL^C]$ and $\P_p[\AL]$. Lower values of these probabilities suggest a smaller estimation variance, hence smaller MSE. The second source of variance-inducing error relates to the selection of the parameter $\ell$, as described in Equations~\eqref{eq:var-est1} and \eqref{eq:var-est2}. Given the limited access to HF function evaluations in one approach (Section~\ref{sssec:ell-one}) or the complete avoidance of HF samples for selecting $\ell$ in another approach (Section~\ref{sssec:ell-two}), a deviation between the chosen $\ell^*$ and the true optimal $\ell$ that minimized Equation~\eqref{eq:bf-is-var-apx-p} inevitably arises. This deviation contributes to an increase in the variance of L-BF-IS estimator and, consequently, its MSE. We acknowledge that fully addressing these challenges, particularly in mitigating bias-inducing and variance-inducing errors, remains an open problem that forms the basis of a future work.

%%%%%%%%%%%%%%%%%%%%%%%%%%%%%%%%%%%%%%%%%%%%%%%%%%%%%%%%%%%%%%%%%%%%%%%%%% Empirical experiments
\section{Empirical Results}
\label{sec:experiments}
This section presents empirical results to illustrate the effectiveness of the L-BF-IS estimator. In Section~\ref{ssec:toy}, a simple 1D function demonstrates the applicability of the MALA on sampling a multi-modal biasing distribution. Then, in Section~\ref{ssec:synthetic}, the L-BF-IS is applied to two different cases: an 8-dimensional Borehole function (detailed in~\ref{sssec:borehole}) and a 1000-dimensional synthetic function (detailed in~\ref{sssec:1000-dim}). The application of the L-BF-IS is shown on two real-world failure probability estimation problems in Section~\ref{ssec:real-prob}, including a composite beam problem (explained in Section~\ref{sssec:beam}) that uses the Euler-Bernoulli equation as an LF model and a steady-state stochastic heat equation (described in Section~\ref{sssec:heat}) with a data-driven LF model based on a pre-trained physics-informed neural operator.

To evaluate the accuracy of our estimations, we use the relative root mean square error (rRMSE), 
\begin{align}
\textup{rRMSE}(N) \coloneqq \sqrt{\mathbb{E}\left[\frac{(\widehat{P}_N-\ExH)^2}{\ExH^2}\right]},
\end{align}
where $\widehat{P}_N$ is the estimator using $N$ iid HF samples. The performance of the L-BF-IS estimator $\widehat{P}^{\text{BF}}_N$ (formulated in Equation~\eqref{eq:bf-is-estimator}) is compared with the standard Monte Carlo estimator $\widehat{P}^{\textup{MC}}_N$ (defined in Equation~\eqref{eq:mc-def}) across all problems. We also produce the LF failure probability, denoted as $\ExL$, which is solely generated from $1\times 10^6$ $h^\LF$ evaluations. For the problems where the input dimension $D \leq 10$, we also consider the results from the Multi-fidelity Importance Sampling (MF-IS) estimator \cite{peherstorfer2016multifidelity}, which uses a biasing distribution created by a Gaussian mixture model. The number of clusters for MF-IS is chosen from $k = \{1, 3, 5, 10\}$, so that the chosen $k$ yields the best performance. We assume the computational costs of HF models are substantially higher than those of the LF models so that the costs of LF forward and derivative evaluations can be ignored. The initial point $\bm z^{(0)}$ of the MALA is typically chosen as the center of the input space. The proposed method requires $\mathcal{O}(M+T+B)$ forward LF model evaluations (typically around $\sim 1\times 10^6$ evaluations) and $\mathcal{O}(N + L)$ forward HF model evaluations, usually between 1 and $\sim 1\times 10^4$.

The experimental component of this study is primarily concerned with scenarios exhibiting a failure probability between $1\%$ and $5\%$. To identify an appropriate LF failure threshold, 
%a threshold is typically set for the HF QoI function $f^{\text{HF}}$, such as $f^\HF(\bm{z})$ is larger than the set threshold, aiming for the observed failure probability to fall in the given range. The HF function $h^{\HF}$ is thus defined to be $h^{\HF}(\bm{z}) < 0$ when the QoI exceeds this predetermined threshold. For the LF function $f^{\LF}$, 
1000 LF QoIs are generated to establish a tentative threshold, ensuring its inducing failure probability is also between $1\%$ and $5\%$ and potentially close to $\ExH$. This procedure is adopted because, for certain LF/HF models (such as the 1000-dimensional problem discussed in Section~\ref{sssec:1000-dim}), there is a notable discrepancy between the values of the LF and HF QoIs. Consequently, applying the same threshold to both models may result in inaccurate probability estimates. In practice, while the HF failure probability $\ExH$ is the goal of estimation, {\it a prior} knowledge of potential range of values is available. Such an estimate is instrumental in establishing a criterion for assessing the utility of LF QoIs within L-BF-IS. 
%Experimental findings suggest that the efficacy of the L-BF-IS approach exhibits limited sensitivity to these assessment criteria, barring substantial deviations.

\subsection{A Simple Bimodal Function for Demonstrating Langevin Algorithm}
\label{ssec:toy}

In failure probability estimation, the multimodal issue occurs when multiple sub-areas in $\Omega$ correspond to failure. The goal of this example is to empirically show that the MALA is capable to address this issue through a 1D example, where
\begin{equation}
h(z) = -(\sin(\pi z)+0.95)(\sin(\pi z) - 0.95).
\end{equation}
The density $p(z)$ is assumed to be uniform between $-1$ and $1$. We choose $\ell=5.0$. The function $h(z)$ in Figure~\ref{fig:toy-function} and the densities $p(z), q(z)$ in Figure~\ref{fig:toy-density} are provided. The biasing density $q^\ast(z)$ shows the bimodal property and we will show that the Langevin algorithm is possible to generate samples from it.

\begin{figure}[ht!]
    \centering
    \begin{minipage}{0.32\textwidth}
        \centering
        \includegraphics[width=\linewidth]{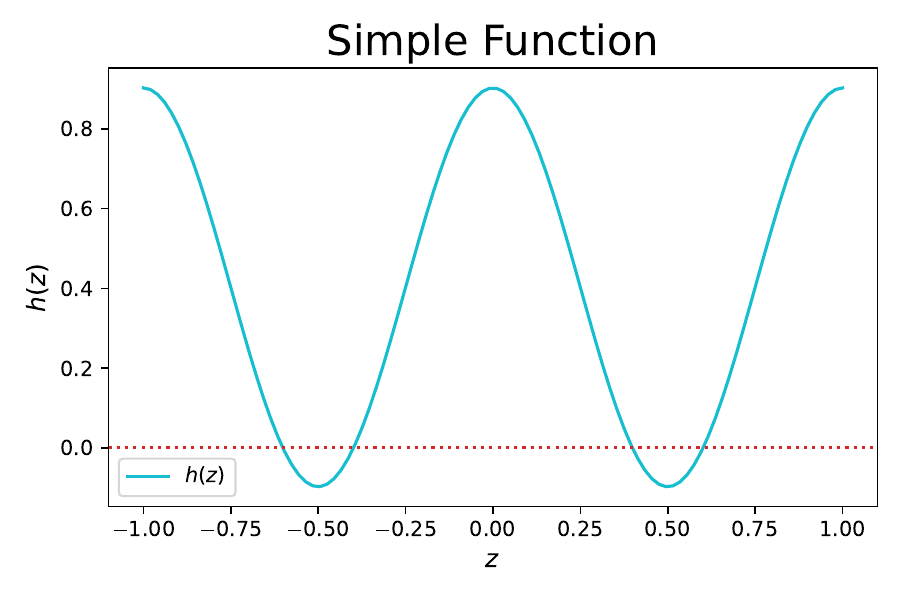}
        \subcaption{}
        \label{fig:toy-function}
    \end{minipage}%
    \hfill
    \begin{minipage}{0.32\textwidth}
        \centering
        \includegraphics[width=\linewidth]{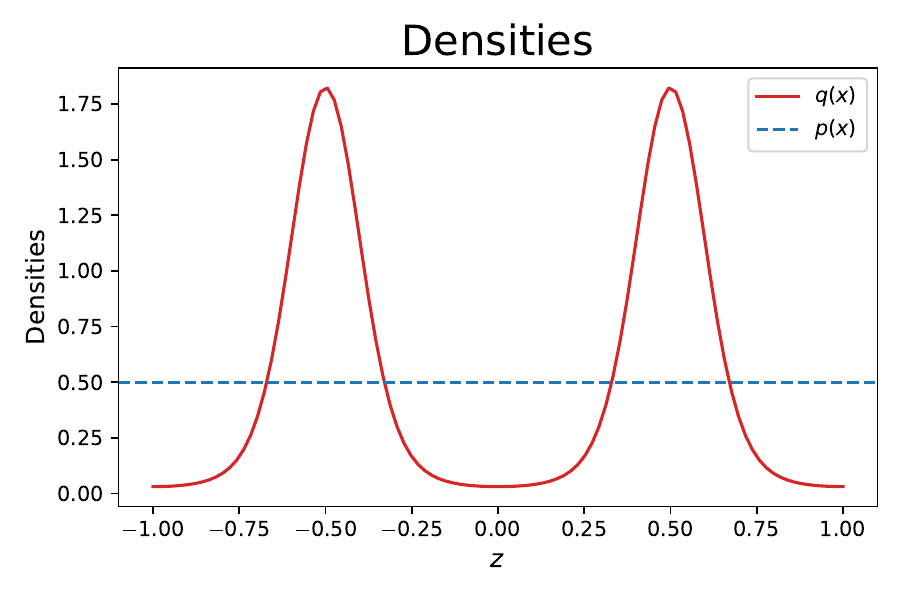}
        \subcaption{}
        \label{fig:toy-density}
    \end{minipage}
    \hfill
    \begin{minipage}{0.32\textwidth}
        \centering
        \includegraphics[width=\linewidth]{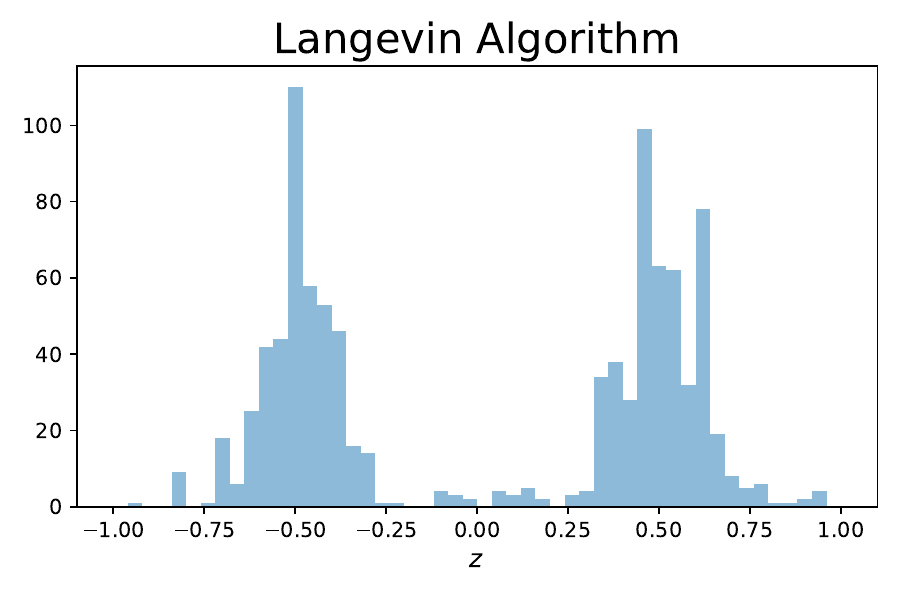}
        \subcaption{}
        \label{fig:toy-sampling}
    \end{minipage}%

    \caption{(a) The example function $h(z)$ and the $0$ threshold. (b) Densities $p(z)$ and $q(z)$ with $\ell=5.0$. (c) Histogram of $1,000$ samples of $q(z)$ generated from the Langevin algorithm described in Algorithm~\ref{alg:langevin-sampling}.}
    \label{fig:toy-case}
\end{figure}

We implement the Langevin algorithm described in Algorithm~\ref{alg:langevin-sampling}, where we set the starting point at $z=0$. The step size $\tau=0.05$ and burn-in number $B=200$. We initiate the Langevin algorithm 100 times and for each chain, we collect 10 samples after the burn-in number. The collected samples are shown in Figure~\ref{fig:toy-sampling}. As we can see, the bimodal shape of the biasing density $q(z)$ is captured by the Langevin algorithm. 

% However, we should note that there is no theoretical guarantee for the convergence of the Langevin algorithm on multimodal functions. The high-dimensional multimodal landscape for failure probability estimation remains a challenging problem to solve, but our experiments in this section demonstrate that the Langevin algorithm has the potential to get close to the final solution.

\subsection{Synthetic Examples with Prescribed Functions}
\label{ssec:synthetic}
\subsubsection{Borehole Function}
\label{sssec:borehole}
We applied the borehole function described in \cite{morris1993bayesian}, which is extended to a multi-fidelity setting in \cite{xiong2013sequential}. It is an 8-dimensional problem that models water flow through a borehole. Following \cite{xiong2013sequential}, the HF QoI function is 
\begin{align}
    f^\HF(\bm z) = \frac{2\pi z_3(z_4-z_5)}{(z_2-\log z_1)\left(1+\frac{2z_7z_3}{(z_2-\log z_1)z_1^2z_8}+\frac{z_3}{z_5}\right)},
\end{align}
and the LF QoI function is 
\begin{align}
    f^\LF(\bm z) = \frac{5z_3(z_4-z_5)}{(z_2-\log z_1)\left(1.5+\frac{2z_7z_3}{(z_2-\log z_1)z_1^2z_8}+\frac{z_3}{z_5}\right)}.
\end{align}
The random inputs $\bm z$ and their distributions are presented in Table~\ref{table:borehole}.
\begin{table}[ht]
	\centering
	\caption{
		The stochastic input ranges, distributions, and physical meanings of the Borehole function.
	}
    \begin{tabular*}{\linewidth}{@{\extracolsep{\fill}} lcl }
	\toprule
	  Range & Distribution & Physical Meaning  \\
	\midrule
	$z_1\in[0.05,0.15]$ &$\mathcal{N}(0.10,0.016)$ 	&  radius of borehole (m)\\
    $z_2\in[4.605,10.820]$ &$\mathcal{N}(7.71,1.0056)$ 	&  radius of influence (m)\\
    $z_3\in[63070,115600]$ & $U[63070,115600]$ 	&  transmissivity of upper aquifer (m$^2$/yr)\\
    $z_4\in[990,1110]$ & $U[990,1110]$ 	&  potentiometric head of upper aquifer (m)\\
    $z_5\in[63.1,116]$ & $U[63.1,116]$ 	&  transmissivity of lower aquifer (m$^2$/yr)\\
    $z_6\in[700,820]$ & $U[700,820]$ 	&  potentiometric head of lower aquifer (m)\\
    $z_7\in[1120,1680]$ & $U[1120,1680]$ 	&  length of borehole (m)\\
    $z_8\in[9855,12045]$ & $U[9855,12045]$ 	&  hydraulic conductivity of borehole (m/yr)\\
	\bottomrule
    \end{tabular*}
	\label{table:borehole}
\end{table}
We define the HF function $h^\HF(\bm{z})$ as $800 - f^\HF(\bm{z})$. To empirically prevent the Langevin Markov chain from moving outside the domain $\Omega$, we introduce an additional penalty term of $100\lVert\bm{z}\rVert^2$ when $\bm{z}$ is outside $\Omega$. Similarly, the LF function $h^\LF(\bm{z})$ is defined as $1000 - f^\LF(\bm{z})$ within the specified domain; otherwise, it takes the penalty term $100\lVert\bm{z}\rVert^2$.  

% %
% \begin{figure}[htp]
%     \centering
%     \begin{minipage}{0.49\textwidth}
%         \includegraphics[width=0.9\textwidth]{Figures/borehole_hist.pdf}
%     \end{minipage}\hfill
%     \begin{minipage}{0.49\textwidth}
%         \includegraphics[width=0.9\textwidth]{Figures/borehole_hist_all.pdf}
%     \end{minipage}
% 	\caption{Histograms showing $1,000$ outputs of the LF function $h^\LF$ (blue) and HF function $h^\HF$ (green) on the left, alongside HF function evaluations on biasing distribution samples from the MALA (orange, $\ell=3.26$) and Gaussian mixture (pink) compared with HF function results on the right. Outputs below $0$ are considered failures.}
% \label{fig:borehole-hist}
% \end{figure}
% %

In Figure~\ref{fig:borehole-ell}, the estimated variance of L-BF-IS estimator using two different approaches for tuning $\ell$ are demonstrated with $L=1\times 10^2$ HF trials and $M=1\times 10^6$ LF evaluations for length scale selection. The uncertainty of the variance estimate is notably higher in the first approach compared to the second, primarily due to the limited number of HF function evaluations available for choosing $\ell$, which significantly raises the likelihood of estimating the variance as zero. 

\begin{figure}[ht!]
	\centering
	\includegraphics[width = 1.0\textwidth]{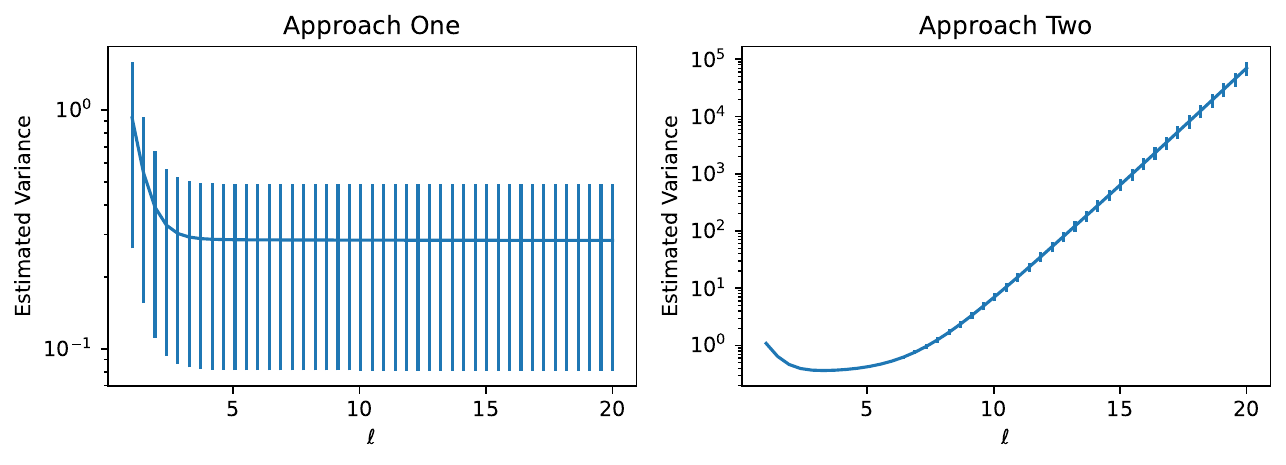}
	\caption{Estimated variance of L-BF-IS for different $\ell$ values with $95\%$ confidence interval using $L=1\times 10^2$ HF evaluations (approach one) and $M=1\times 10^6$ LF evaluations (approach two) for the borehole function in Section~\ref{sssec:borehole}. Approach one exhibits higher estimation uncertainty, whereas approach two is more robust.}
	\label{fig:borehole-ell}
\end{figure}
\begin{figure}[ht!]
    \begin{minipage}{0.47\textwidth}
        \includegraphics[width=1.0\textwidth]{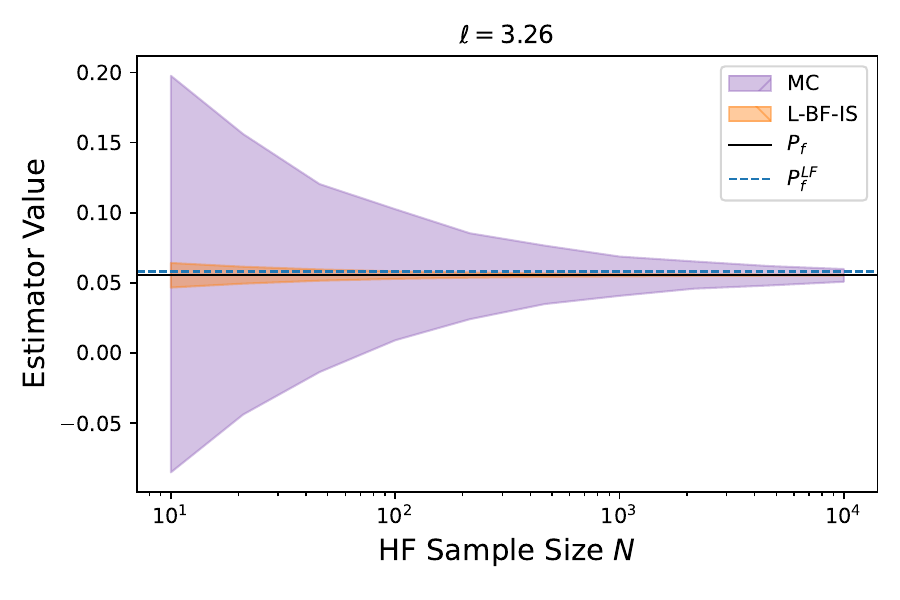}
        \subcaption{}
        \label{fig:borehole-ell326values}
    \end{minipage}
    \hfill
    \begin{minipage}{0.47\textwidth}
        \includegraphics[width=1.0\textwidth]{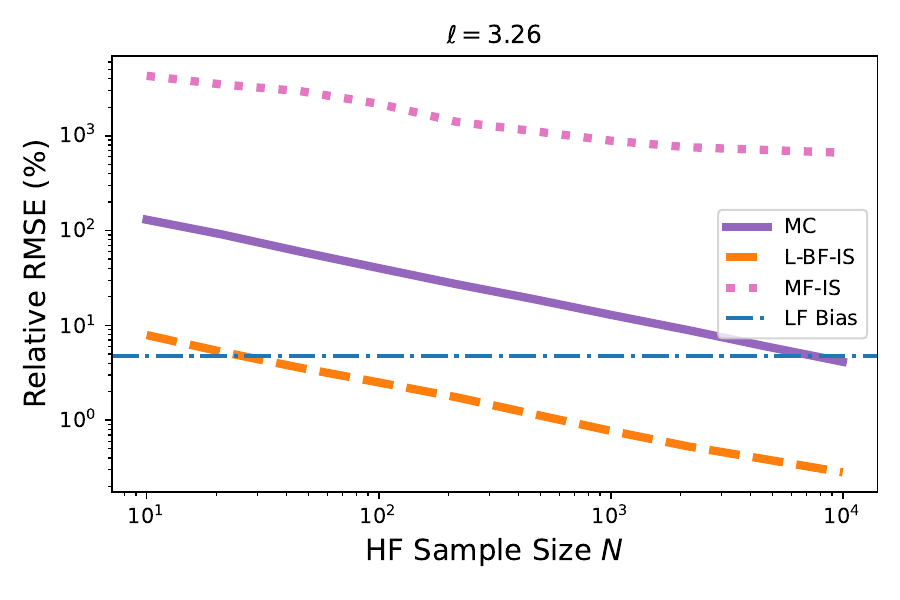}
        \subcaption{}
        \label{fig:borehole-ell326rrmse}
    \end{minipage}
    
    \begin{minipage}{0.47\textwidth}
        \includegraphics[width=1.0\textwidth]{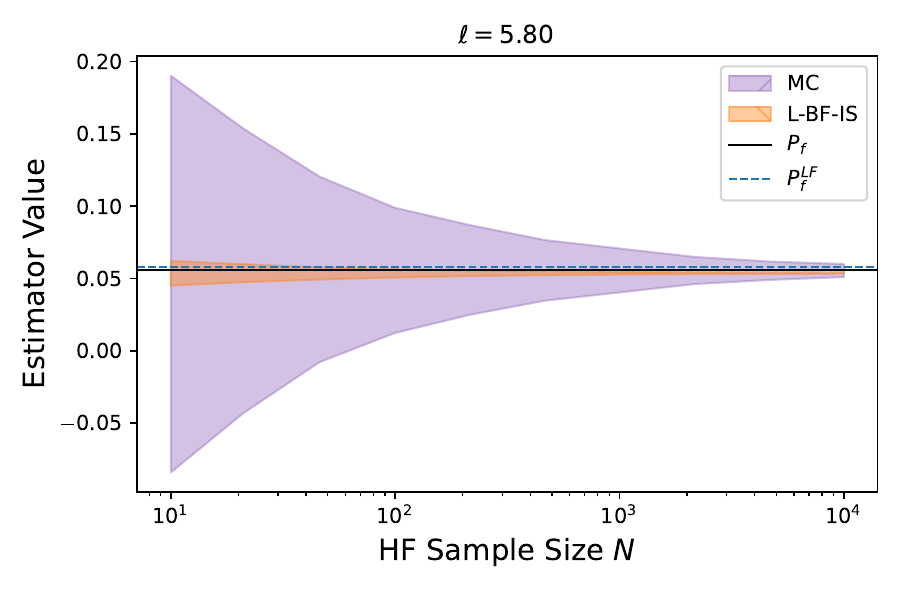}
        \subcaption{}
        \label{fig:borehole-ell580values}
    \end{minipage}
    \hfill
    \begin{minipage}{0.47\textwidth}
        \includegraphics[width=1.0\textwidth]{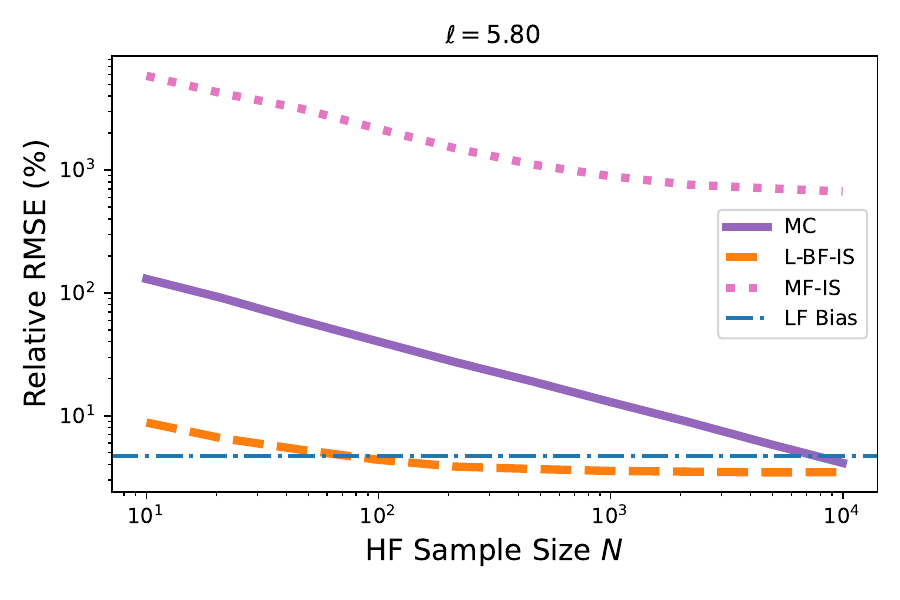}
        \subcaption{}
        \label{fig:borehole-ell580rrmse}
    \end{minipage}

    \begin{minipage}{0.47\textwidth}
        \includegraphics[width=1.0\textwidth]{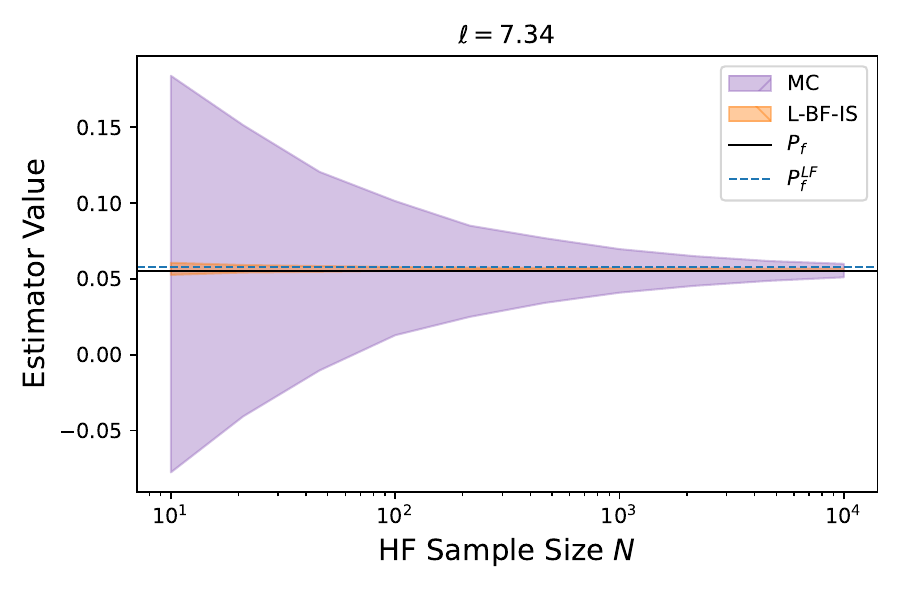}
        \subcaption{}
        \label{fig:borehole-ell734values}
    \end{minipage}
    \hfill
    \begin{minipage}{0.47\textwidth}
        \includegraphics[width=1.0\textwidth]{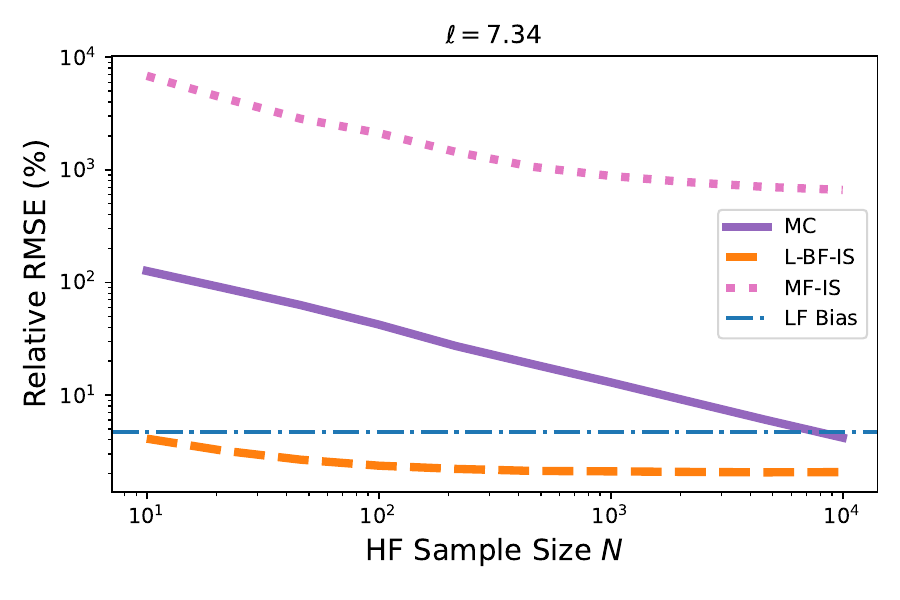}
        \subcaption{}
        \label{fig:borehole-ell734rrmse}
    \end{minipage}
    \caption{Convergence behavior of L-BF-IS (dash) for $\ell$ values of 3.26 (a-b), 5.80 (c-d), and 7.34 (e-f), compared with standard Monte Carlo (solid), MF-IS (dot), and LF failure probability (dash dot) using 10 Gaussian mixture clusters for the borehole function in Section~\ref{sssec:borehole}. The blue dash dotted lines are LF failure probabilities. The shaded areas represent the $95\%$ confidence interval from $1,000$ trials.}
    \label{fig:borehole}
\end{figure}

To demonstrate the robustness of L-BF-IS with respect to the choice of $\ell$, we compare the convergence results of L-BF-IS with standard Monte Carlo and MF-IS across three different $\ell$ values in Figure~\ref{fig:borehole}. The MALA step size $\tau$ is set to $1\times 10^{-4}$, with a burn-in value of $B = 1\times 10^3$ and an iteration number $T=1\times 10^4$. For the convergence analysis, the estimates are computed for HF sample size $N$ as $10,21,46,100,215,464,1000,2154,4641$, and $10000$ across $1000$ trials to determine the $95\%$ confidence intervals. In this example, the LF model produces similar results to the HF model, with $5\%$ relative error in estimating $\ExH$. When $\ell$ is set to $3.26$ (following approach two), the L-BF-IS successfully reduces the relative RMSE to $0.3\%$. However, when $\ell$ is not optimally chosen, as shown in Figure~\ref{fig:borehole-ell580rrmse} with $\ell=5.80$ or Figure~\ref{fig:borehole-ell734rrmse} with $\ell=7.34$, the improvements are limited to $5\% \sim 8\%$. 

We also investigate the performance of L-BF-IS when the value of $\ExH$ is smaller and the LF model is less accurate. We choose the new LF and HF functions as $h^\LF(\bm z) = 1100 - f^\LF(\bm z)$ and $h^\HF(\bm z) = 900 - f^\HF(\bm z)$, respectively. With a smaller value of failure probability, the region that the biasing distribution should place more probabilities becomes smaller. With updated LF and HF functions, the value of $\ell$ is chosen as $3.71$ using approach two, and the corresponding convergence results are presented in Figure~\ref{fig:borehole-lp}. The relative RMSE of the LF model is $36\%$,  which is significantly larger than the previous case. We notice that the relative RMSE of L-BF-IS maintains its quality and is $1\%$, which is one order of magnitude better than the standard Monte Carlo method on the HF function.

\begin{figure}[ht!]
    \begin{minipage}{0.48\textwidth}
    \includegraphics[width=1.0\textwidth]{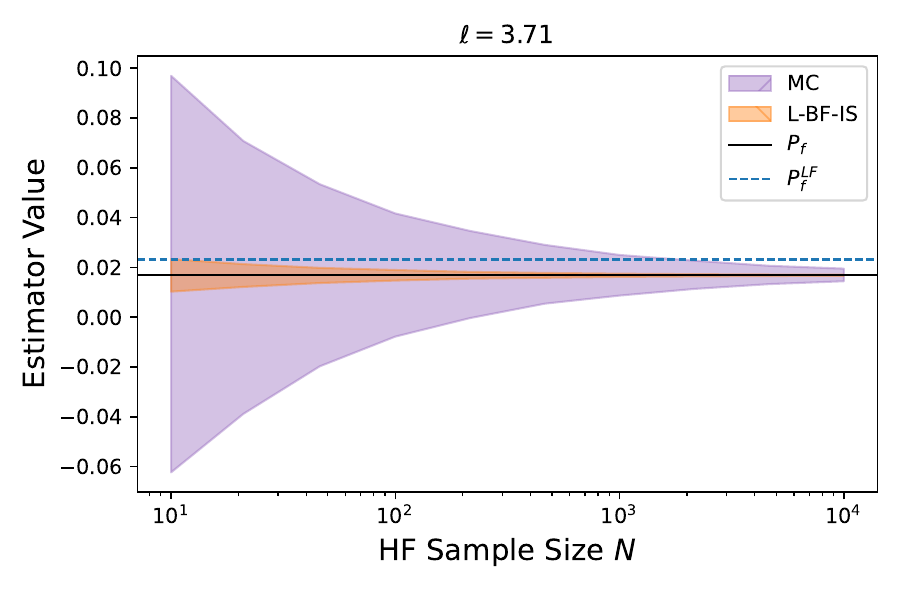}
    \subcaption{}
    \end{minipage}
    \begin{minipage}{0.48\textwidth}
    \includegraphics[width=1.0\textwidth]{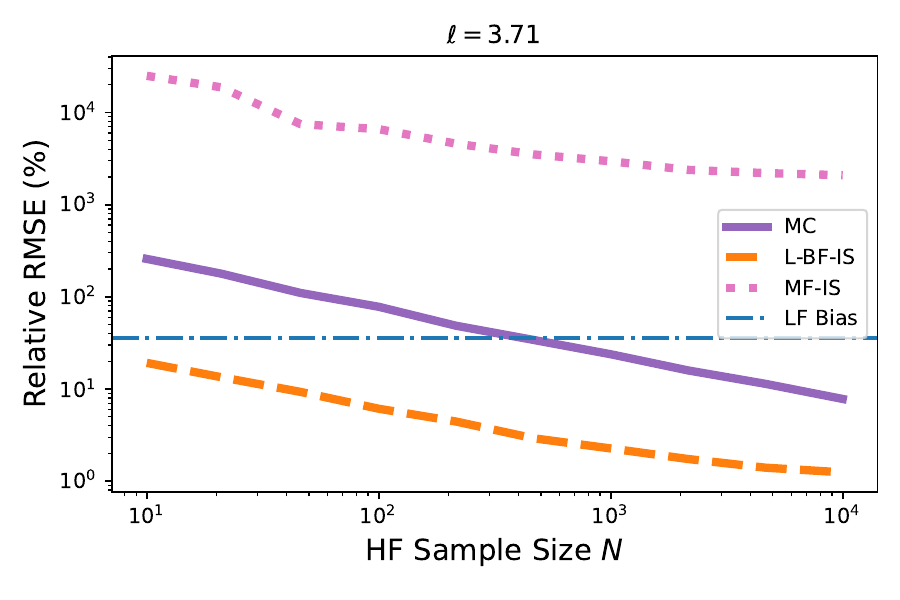}
    \subcaption{}
    \end{minipage}

    \caption{Convergence behavior of L-BF-IS (dash) for $\ell = 3.71$ compared with standard Monte Carlo (solid) and LF failure probability (dash dot) with updated LF and HF functions for the borehole function in Section~\ref{sssec:borehole}. The shaded areas represent the $95\%$ confidence interval from $1,000$ trials.}
    \label{fig:borehole-lp}
\end{figure}
\subsubsection{1000 Dimensional Synthetic Function}
\label{sssec:1000-dim}
To evaluate the performance of L-BF-IS on high-dimensional problems, we examine a 1000-dimensional problem following \cite{hampton2018basis}. The HF QoI function is defined as
\begin{align}
    f^\HF(\bm z) = \exp\left(2-\sum_{k=1}^{1000}\frac{\sin(k)z_k}{k}\right)
\end{align}
and the LF QoI function the truncated Taylor series expansion,
\begin{equation}
    f^\LF(\bm z) = \sum_{m=0}^2 (m!)^{-1}\left(2-\sum_{k=1}^{1000}\frac{\sin(k)z_k}{k}\right)^m.
\end{equation}
The HF function $h^\HF(\bm z)$ is set to $20 - f^\HF(\bm z)$ within the hypercube $[-1, 1]^{1000}$ domain, and we apply a penalty of $100\lVert\bm z\rVert^2$ outside this domain. Similarly, the LF function $h^\LF(\bm z)$ is defined as $8 - f^\LF(\bm z)$ with the same penalty applied. 

% %
% \begin{figure}[H]
%     \centering
%     \begin{minipage}{0.49\textwidth}
%         \includegraphics[width=0.9\textwidth]{Figures/onek_hist.pdf}
%     \end{minipage}\hfill
%     \begin{minipage}{0.49\textwidth}
%         \includegraphics[width=0.9\textwidth]{Figures/onek_hist_all.pdf}
%     \end{minipage}
% 	\caption{Histograms for the 1000-dimensional function, similar to Figure~\ref{fig:borehole-hist}, showcasing LF and HF function outputs (left) and HF function evaluations on the biasing distribution with $\ell=2.81$ (right).}
% \label{fig:onek-hist}
% \end{figure}
% %

For tuning the value of $\ell$, we employed two approaches, utilizing $L=1\times 10^2$ HF trial evaluations (for approach one) and $M=1\times 10^6$ LF evaluations (for both methods) for variance estimation. These calculations were repeated ten times to estimate their variability. Unlike the borehole example in Section~\ref{sssec:borehole}, both approaches yielded similar variance estimates for this high-dimensional problem, though approach one exhibited larger variability, as depicted in Figure~\ref{fig:onek-ell}.

\begin{figure}[ht!]
	\centering
	\includegraphics[width = 1.0\textwidth]{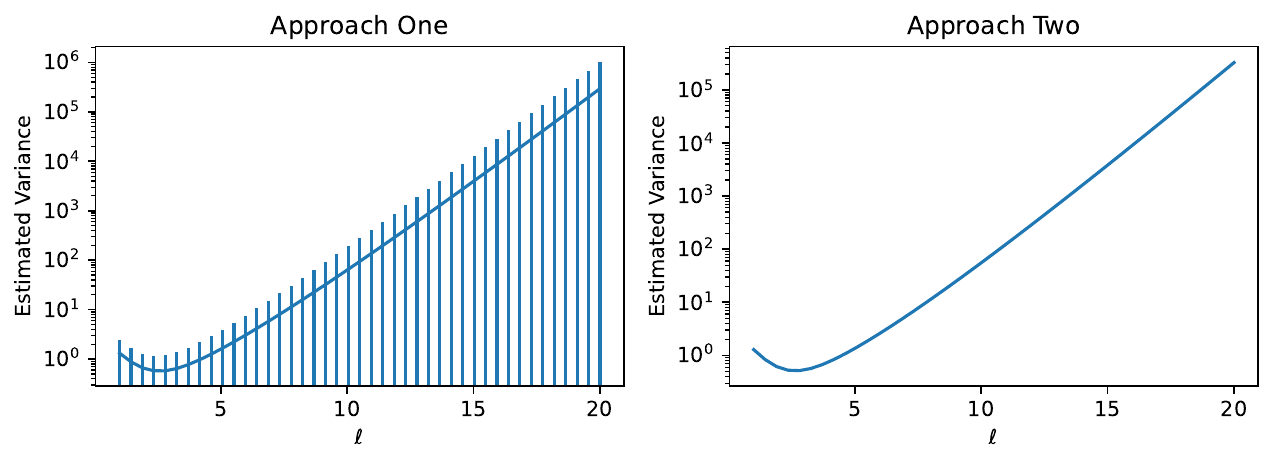}
	\caption{Estimated variance of L-BF-IS across different $\ell$ values, with uncertainty bars indicating a $95\%$ confidence interval. Estimates are based on $L=1\times 10^2$ HF evaluations (approach one) and $M=1\times 10^6$ LF evaluations (both approaches).}
	\label{fig:onek-ell}
\end{figure}

Given the consistent results in Figure~\ref{fig:onek-ell}, we selected an $\ell$ value of $2.36$ for this problem. Due to the high-dimensionality of this problem, we compared the convergence of L-BF-IS solely with the Monte Carlo method. The MALA step size $\tau$ is set to $1\times 10^{-5}$, with a burn-in number $B$ of $1\times 10^4$ and $T=1\times 10^4$ iterations. The L-BF-IS is compared with MC estimator with HF sample sizes $N$ as $10, 21, 46, 100, 215, 464, 1000, 2154, 4641$, and $10000$ across $1000$ trials to calculate the $95\%$ confidence intervals. The failure probability produced by the LF model $\ExL$ has relative RMSE of around $64\%$, while the L-BF-IS is able to reduce it to around $20\%$. However, the convergence outcomes in Figure~\ref{fig:onek} reveal a bias of $2\%$ in the L-BF-IS estimate, which we attribute to the Langevin algorithm's inaccuracies discussed in Section~\ref{ssec:error}. Despite this bias, L-BF-IS still offers a significant improvement of the MSE for smaller HF sample sizes ($N\leq 300$).
\begin{figure}[ht!]
    \begin{minipage}{0.48\textwidth}
    \includegraphics[width=1.0\textwidth]{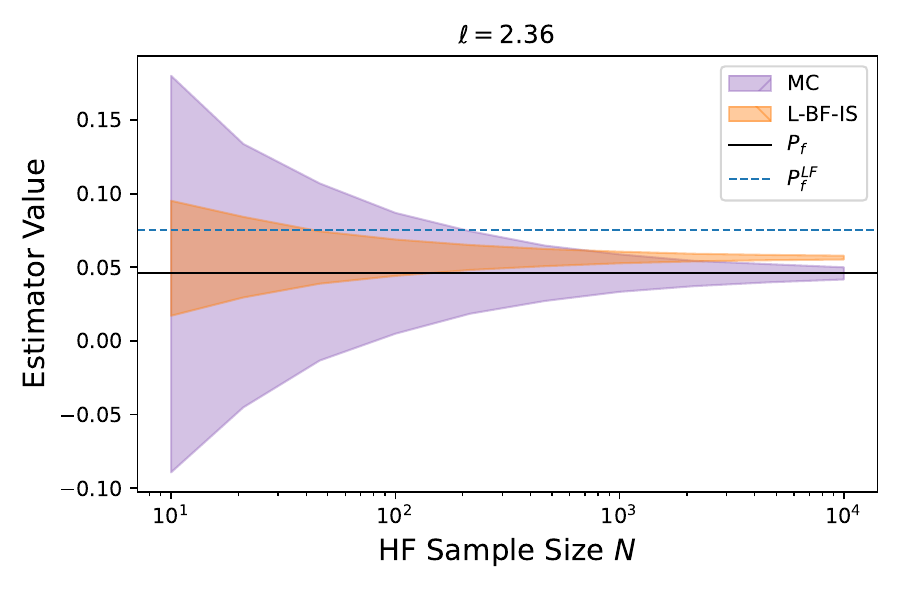}
    \subcaption{}
    \end{minipage}
    \begin{minipage}{0.48\textwidth}
    \includegraphics[width=1.0\textwidth]{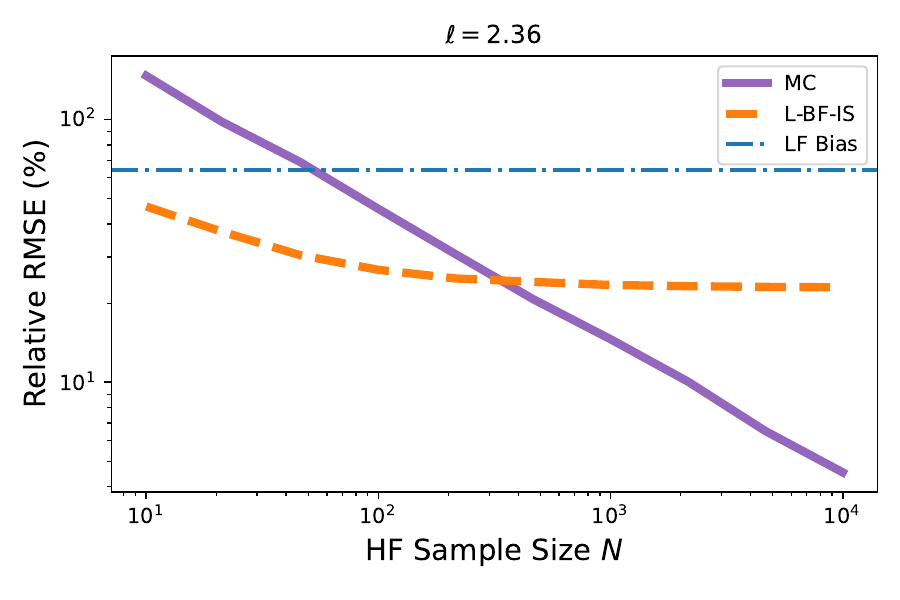}
    \subcaption{}
    \end{minipage}

    \caption{Convergence of L-BF-IS (dash) for selected $\ell=2.36$ value compared with standard Monte Carlo (solid) and LF failure probability (dash dot) for the $1000D$ problem in Section~\ref{sssec:1000-dim}.}
    \label{fig:onek}
\end{figure}
%

%%%%%%%%%%%%%%%%%%%%%%%%%%%%%%%%%%%%%%%%%%
\subsection{Physics-based Examples}
\label{ssec:real-prob}

\subsubsection{Composite Beam}
\label{sssec:beam}

Building on the work of \cite{hampton2018practical,de2020transfer,de2022neural,cheng2024subsampling,cheng2024bi}, we examine a plane-stress, cantilever beam featuring a composite cross-section and hollow web, as depicted in Figure \ref{fig:cantilever}. The focus is on the maximum displacement of the top cord, with uncertain parameters $z_1, z_2, z_3, z_4$. Here, $z_1$ represents the intensity of the distributed force applied to the beam, while $z_2$, $z_3$, and $z_4$ denote Young's moduli of the cross-section's three components. These parameters are independent and uniformly distributed, with the input parameter dimension being $D = 4$. The QoI of this problem is the maximum displacement at the top of the beam. Table \ref{table:beam} outlines the range of input parameters alongside other deterministic parameters.

\begin{figure}[htbp]
	\centering
	\includegraphics[trim = 65mm 65mm 65mm 65mm, clip, width = 0.6\textwidth]{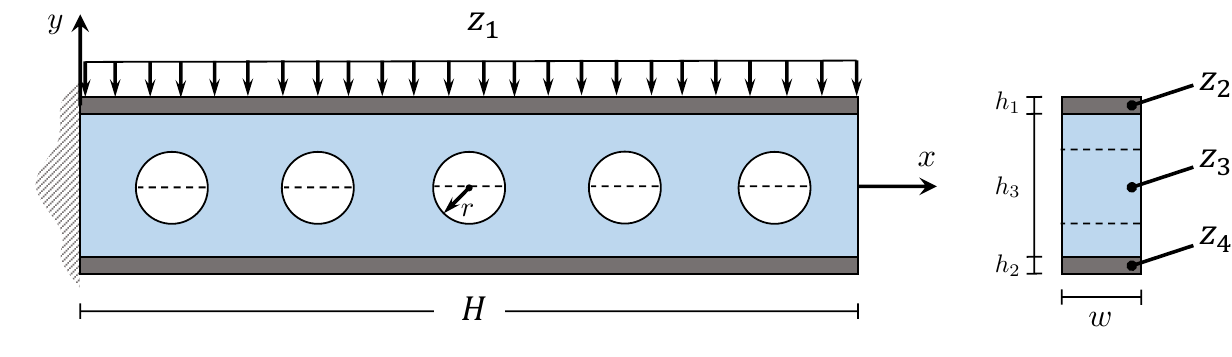}\\
	\includegraphics[trim = 0mm 80mm 0mm 70mm, clip, width = 0.7\textwidth]{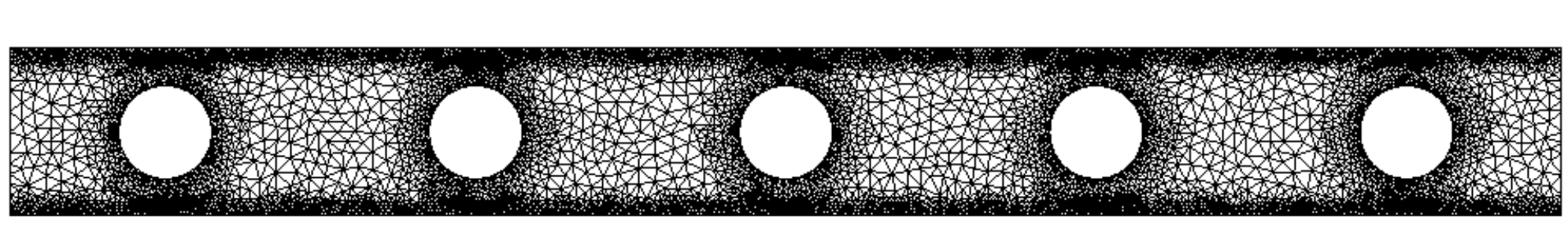}
	\caption{Top: Cantilever beam (left) and the composite cross section (right) adapted from \cite{hampton2018practical}. Bottom: Finite element mesh used to generate high-fidelity solutions.}
	\label{fig:cantilever}
\end{figure}

\begin{table}[htbp]
	\centering
	\caption{The parameter values in the composite cantilever beam model. 
		The center of the holes are at $x=\{5,15,25,35,45\}$.
		The parameters $z_1$, $z_2$, $z_3$ and $z_4$ are drawn independently and uniformly at random from the specified intervals.
		\label{table:beam}} 
	\centering
        \resizebox{0.8\textwidth}{!}{
	\begin{tabular}{ c c c c c c c c c c}   
		\toprule 
			$H$ & $h_1$ & $h_2$ & $h_3$ & $w$ & $r$ & $z_1$ & $z_2$ & $z_3$ & $z_4$ \\
		\midrule
			50  &  0.1 & 0.1 & 5 &  1 &  1.5 & $[9,11]$ & $[0.9\text{e6}, 1.1\text{e6}]$  &  $[0.9\text{e6}, 1.1\text{e6}]$ & $[0.9\text{e4}, 1.1\text{e4}]$  \\ 
		\bottomrule
	\end{tabular} 
      }
	\label{table:parameters} 
\end{table}

This study employs two models to represent the HF and LF QoI functions. The HF QoI function $f^\HF$ is obtained from a finite element analysis using a triangular mesh. In contrast, the LF QoI function $f^\LF$ is evaluated based on the Euler-Bernoulli beam theory, which simplifies the model by ignoring shear deformation and circular holes. The Euler-Bernoulli theorem provides a differential equation for vertical displacement $u(x)$, which can be explicitly solved as 
\begin{equation}\label{eq:Euler-Bernoulli}
  EI\frac{d^4u(x)}{dx^4}=-z_1 \hskip 10pt \Longrightarrow \hskip 10pt
u(x)=-\frac{z_1 H^4}{24EI}\left(\left(\frac{x}{H}\right)^4-4\left(\frac{x}{H}\right)^3+6\left(\frac{x}{H}\right)^2\right),
\end{equation}
where $E$ and $I$ represent Young's modulus and the moment of inertia, respectively. We take $E=z_4$, and the width of the top and bottom sections are $w_1=(z_2/z_4)w$ and $w_2=(z_3/z_4)w$ respectively, while all other dimensions are the same as Figure \ref{fig:cantilever} shows. 
For simulation convenience, we generate $10,000$ realizations from both HF and LF QoI functions, constructing their surrogates $\tilde{f}^\HF$ and $\tilde{f}^\LF$ using polynomial chaos expansion (PCE) with a total degree of 3. The relative MSE of $\tilde{f}^\HF$ and $\tilde{f}^\LF$ are $1.38\times 10^{-2}$ and $1.26\times 10^{-2}$, respectively. We define the HF function $h^\HF(\bm z)\coloneqq \tilde{f}^\HF(\bm z) + 4.04$ and LF function $h^\LF(\bm z) \coloneqq \tilde{f}^\LF(\bm z) + 3.18$. These two functions output negative values if the displacement of the composite beam is less than $-4.04$ (for HF) or $-3.18$ (for LF). These two functions are defined so that negative values indicate failures.

To determine the optimal $\ell$, we replicated the experimental setup used in the previous examples, employing $L=1\times 10^2$ HF trials and $M=1\times 10^6$ LF trials, with the process repeated ten times to account for uncertainty. The estimated variances for various $\ell$ values are illustrated in Figure~\ref{fig:beam-ell}, showing slight differences between the two approaches. The values of $\ell$ that minimize the variance are $14.90$ and $18.57$ for approach one and approach two, respectively. Based on these findings, we proceed with the convergence analysis using the two identified $\ell$ values.

\begin{figure}[ht!]
	\centering
	\includegraphics[width = 1.0\textwidth]{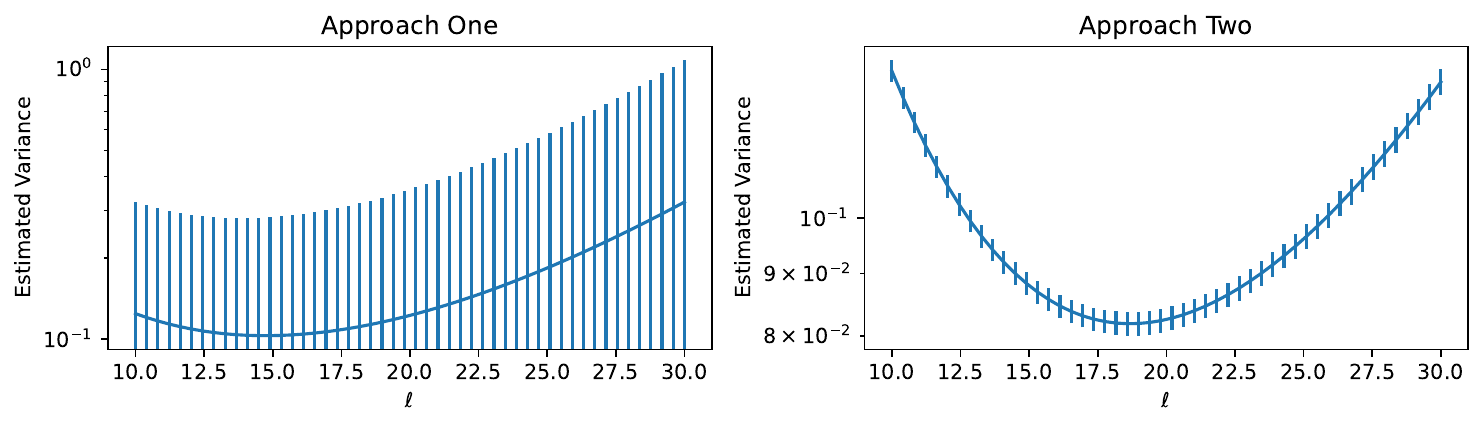}
	\caption{Estimated variance of L-BF-IS for different $\ell$ values, with uncertainty bars representing the $95\%$ confidence interval. Using $L=100$ HF evaluations (approach one) and $M=1,000,000$ LF evaluations (both approaches).}
	\label{fig:beam-ell}
\end{figure}
%

% %
% \begin{figure}[htp]
%     \centering
%     \begin{minipage}{0.49\textwidth}
%         \includegraphics[width=0.9\textwidth]{Figures/beam_hist.pdf}
%     \end{minipage}\hfill
%     \begin{minipage}{0.49\textwidth}
%         \includegraphics[width=0.9\textwidth]{Figures/beam-hist-all.pdf}
%     \end{minipage}
% 	\caption{Histograms for the 1000-dimensional function, similar to Figure~\ref{fig:beam-hist} and Figure~\ref{fig:borehole-hist}, showcasing LF and HF function outputs (left) and HF function evaluations on the biasing distribution from the MALA with $\ell = 14.90$ and Gaussian mixture with 10 clusters (right).}
% \label{fig:beam-hist}
% \end{figure}
% %
In this study, we set the MALA step size $\tau = 1 \times 10^{-3}$ and chose both the burn-in number $B$ and iteration number $T$ to be $10,000$. The starting value $\bm z^{(0)}$ is $[1\times 10^1, 1\times 10^6, 1\times 10^6, 1\times 10^4]$. This approach was compared with the MF-IS technique \cite{peherstorfer2016multifidelity}, which uses a Gaussian mixture model with 10 cluster centers. For the convergence analysis, we used HF sample sizes $N$ as $10, 21, 46, 100, 215, 464, 1000, 2154, 4641$, and $10000$ with experiments repeated $1,000$ times to assess the standard deviation of the results. The outcomes, illustrated in Figure~\ref{fig:beam}, highlight a noteworthy observation regarding the impact of an inaccurately chosen parameter $\ell$ on the L-BF-IS estimates. Specifically, with $\ell = 18.57$, the biasing distribution derived from the MALA significantly reduces the RMSE at the early stage, whereas the convergence with $\ell = 14.90$ demonstrated relatively inferior performance. Given that $\ell = 14.90$ was obtained using approach one and $\ell = 18.57$ using approach two, our findings suggest that the latter provides a more accurate determination of the optimal $\ell$ value. Note that the $L$ HF evaluations for approach one is not included in this figure. A potential reason for the observed bias of $\ell=14.90$ is that varying $\ell$ values alter the smoothness conditions of the resultant biasing densities, which in turn negatively affects the convergence performance of the Langevin algorithm.

\begin{figure}[ht!]
    \begin{minipage}{0.48\textwidth}
    \includegraphics[width=1.0\textwidth]{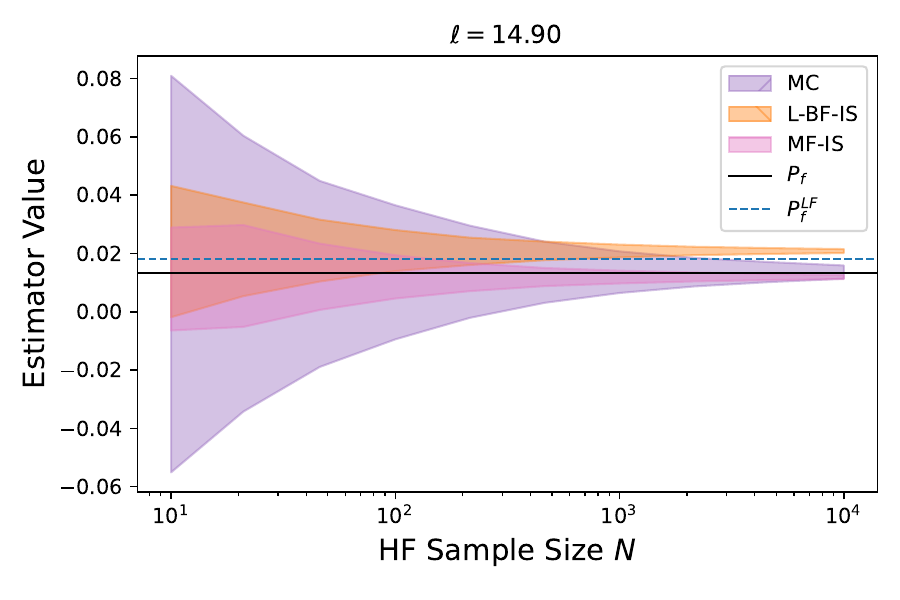}
    \subcaption{}
    \end{minipage}
    \hfill
    \begin{minipage}{0.48\textwidth}
    \includegraphics[width=1.0\textwidth]{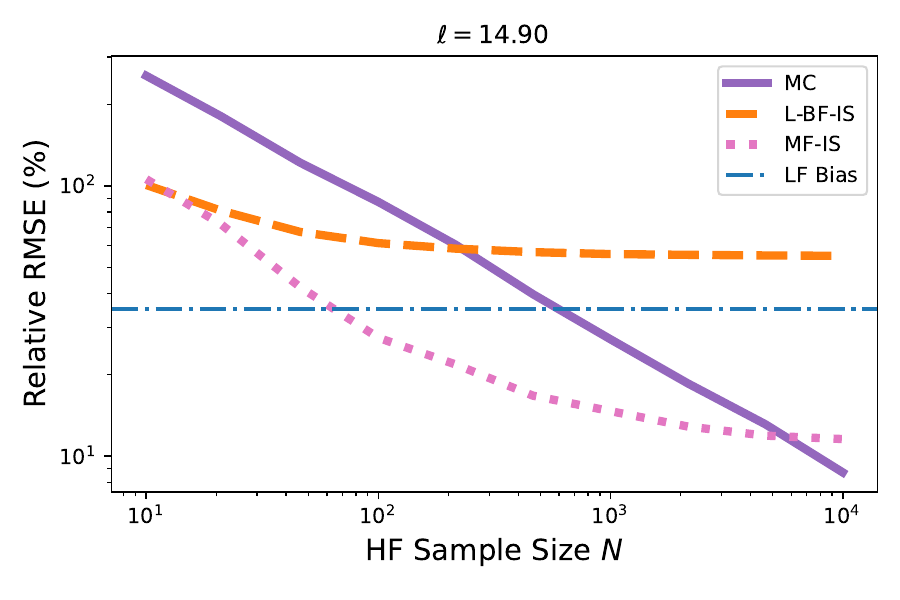}
    \subcaption{}
    \end{minipage}
    \begin{minipage}{0.48\textwidth}
    \includegraphics[width=1.0\textwidth]{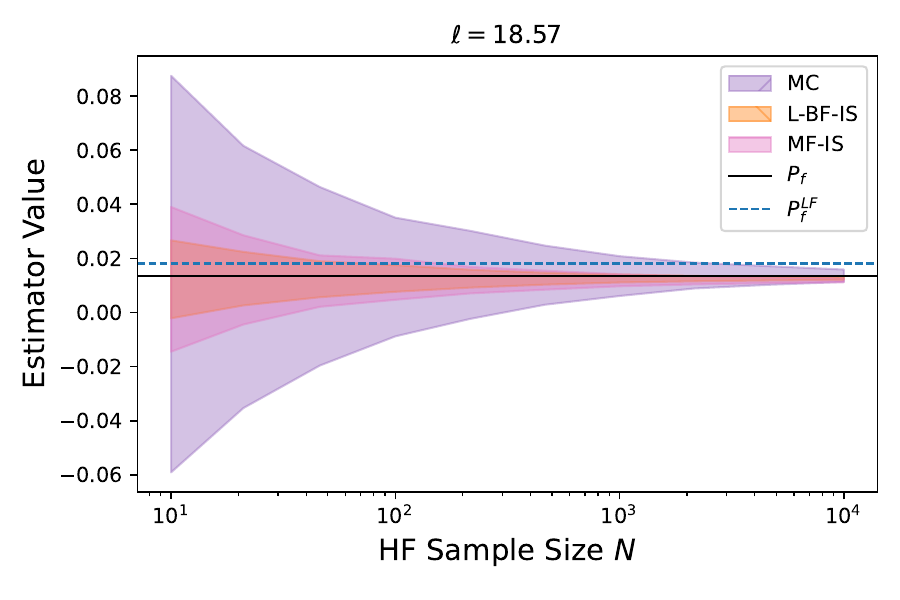}
    \subcaption{}
    \end{minipage}
    \hfill
    \begin{minipage}{0.48\textwidth}
    \includegraphics[width=1.0\textwidth]{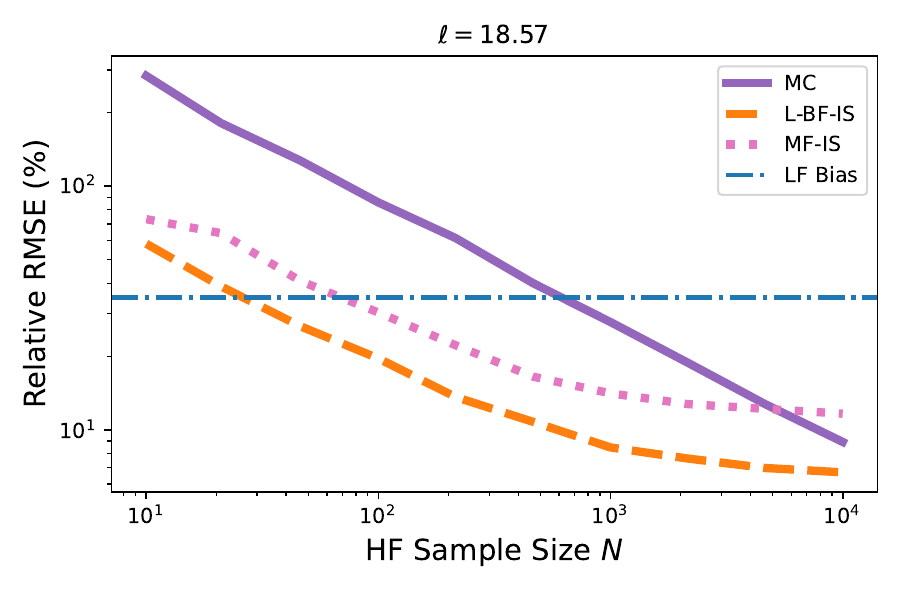}
    \subcaption{}
    \end{minipage}
    \caption{Convergence behavior of L-BF-IS (dash) for $\ell$ values of 14.90 (a-b) and 18.57 (c-d), compared with standard Monte Carlo (solid) and MF-IS (dot) using 10 Gaussian mixture clusters for the beam problem in Section~\ref{sssec:beam}. The shaded areas represent the $95\%$ confidence interval from $1,000$ trials.}
    \label{fig:beam}
\end{figure}

\subsubsection{Steady-state Heat Equation with Random Inputs}
\label{sssec:heat}

In this section, we discuss the performance of a 2D steady-state stochastic heat equation, with uncertain thermal coefficient $K$. The steady-state heat equation can be described as
\begin{equation}
\label{eqn:steady-state}
\begin{aligned}
&-\frac{\partial}{\partial \bm x}\left(K(\bm x,\bm z)\frac{\partial u(\bm x, \bm z)}{\partial \bm x}\right) = 1.0,\quad \bm x\in (0, 1)^2\\
& u(\bm x,\bm z) = 0,\quad x_1 \in\{0, 1\} \text{ or } x_2\in\{0, 1\}.
\end{aligned}
\end{equation}
The thermal coefficient $K(\bm x, \bm z)$ is defined as a stochastic process given by \cite{rahimi2007random},
\begin{equation}
\label{eq:thermal-coef}
K(\bm x, \bm z) = \bar{K} + \exp\left(\frac{\sqrt{2}}{\sqrt{D'}}\sum_{i=1}^{D'} z^w_i\cos\left(z^{a_1}_i x_1 + z^{a_2}_i x_2 + z^b_i\right)\right),
\end{equation}
where $\bar{K}=3$, $z^w_i, z^{a_1}_i, z^{a_2}_i\overset{iid}{\sim}\mathcal{N}(0, 1)$ and $z^b_i\overset{iid}{\sim}U[0,2\pi]$. The corresponding covariance kernel function of the Gaussian process for the exponent part of Equation~\eqref{eq:thermal-coef} is $k(x_1,x_2)=\exp(-(x_1-x_2)^2)$. The dimension of the problem is $4D' = D = 400$. The QoI is the solution on the domain $(0, 1)^2$ with a grid size of $\Delta x_1 = \Delta x_2 = 1.67\times 10^{-2}$ in each direction, then leading to a solution in $\R^{61\times 61}$. Therefore the QoI functions are defined as $f^\LF, f^\HF:\R^{400}\to\R^{61\times 61}$. The finite difference method is used to compute the HF QoI and a pre-trained Physics-informed Neural Operator (PINO) \cite{li2021physics} as a surrogate LF QoI function. The core idea of PINO is to construct a deep-learning-based surrogate that learns the operator $\mathcal{G}$, such that $\mathcal{G}(K)(\bm{x}, \bm{z}) \approx u(\bm{x}, \bm{z})$. We use a PINO model pre-trained following \cite{li2021physics}; however, since the distribution of $K$ in \cite{li2021physics} differs from our $K$ defined in Equation~\eqref{eq:thermal-coef}, this setting can be treated as a transfer learning problem. In this case, the training data for PINO are not counted as additional HF evaluations. The deep-learning structure of the PINO provides the Jacobian $\partial\mathcal{G}(K)(\bm x, \bm z)/\partial \bm z$ given $\bm x$ is defined over a fixed grid. In Figure~\ref{fig:darcy-examples}, three samples of $K(\bm x, \bm z)$ and the corresponding realizations of $u(\bm x, \bm z)$ are presented. 

% %
% \begin{figure}[ht!]
%     \centering
%     \includegraphics[width=1.0\textwidth]{Figures/darcy-pino-fd.pdf}
%     \caption{The solutions of the steady-state heat equation in Equation~\eqref{eqn:steady-state} given three different realizations of the thermal coefficients $K(\bm x, \bm z)$ on a $61\times 61$ grid of $(0, 1)^2$ sampled from Equation~\eqref{eq:thermal-coef}. The left column is the 2D visualization of the thermal coefficient, the middle column is the LF QoI solution provided by a pre-trained PINO, and the right column is the HF QoI solution provided by the finite difference method.}
%     \label{fig:darcy-examples}
% \end{figure}
% %

We assume the system fails if the maximum value of $u(\bm x, \bm z)$ is larger than some thresholds. The LF and HF functions are defined as $h^\LF(\bm z) = 0.019 - \max(f^\LF)$ and $h^\HF(\bm z) = 0.022 - \max(f^\HF)$, respectively. We choose $M=1\times 10^6$ and $L=1\times 10^2$ for selecting $\ell$, with the result presented in Figure~\ref{fig:darcy-ell}. The optimal value of $\ell$, as selected by the approach two, is $1786.0$.

\begin{figure}[htp]
    \centering
    \includegraphics[width=1.0\textwidth]{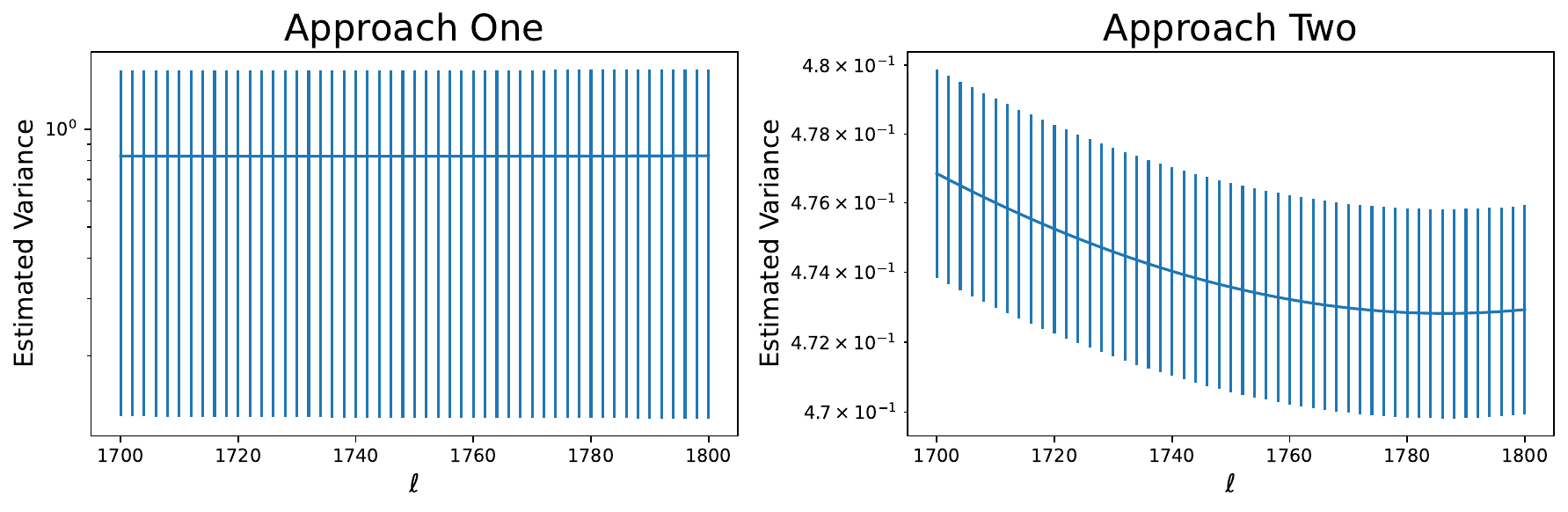}
    \caption{The estimated variance of L-BF-IS with $95\%$ confidence intervals across varying values of $\ell$ is illustrated in the left figure using approach one and in the right figure using approach two. It is worth noting that the left figure exhibits a minimum point; however, the uncertainty is sufficiently large to obscure its depiction.}
    \label{fig:darcy-ell}
\end{figure}

The Langevin algorithm is employed with step size $\tau=1\times 10^{-4}$ and burn-in number $B=1\times 10^4$. We provide three examples of the thermal coefficients that are sampled from the biasing density $q(\bm z)$ with the associated LF and HF QoIs presented in Figure~\ref{fig:darcy-q-examples}. Comparing Figure~\ref{fig:darcy-examples} and Figure~\ref{fig:darcy-q-examples}, we note that the results generated from the Langevin algorithm are more likely to produce failure results defined by functions $h(\cdot)$ and with smaller variance relative to the original reference density $p(\bm z)$. These examples explain why importance sampling using the Langevin algorithm can help reduce the variance of the estimates and eventually the MSE.

% %
% \begin{figure}[ht!]
%     \centering
%     \includegraphics[width=1.0\textwidth]{Figures/darcy-pino-fd-q.pdf}
%     \caption{The solutions of the steady-state heat equation in Equation~\eqref{eqn:steady-state} given three different realizations of the thermal coefficients $K(\bm x, \bm z)$ on a $61\times 61$ grid of $(0, 1)^2$ sampled from the biasing density $q(\bm z)$ with Langevin algorithm. The left column is the 2D visualization of the thermal coefficient, the middle column is the LF QoI solution provided by a pre-trained PINO, and the right column is the HF QoI solution provided by the finite difference method. Compared with examples in Figure~\ref{fig:darcy-examples}, the QoIs sampled from $q(\bm z)$ are more likely to have extreme maximum values ($h^\HF(\bm z) \geq 0.022$) and reduced variance.}
%     \label{fig:darcy-q-examples}
% \end{figure}
% %

\begin{figure}[ht!]
    \centering
    \begin{minipage}[b]{0.49\textwidth}
        \centering
        \includegraphics[width=\textwidth]{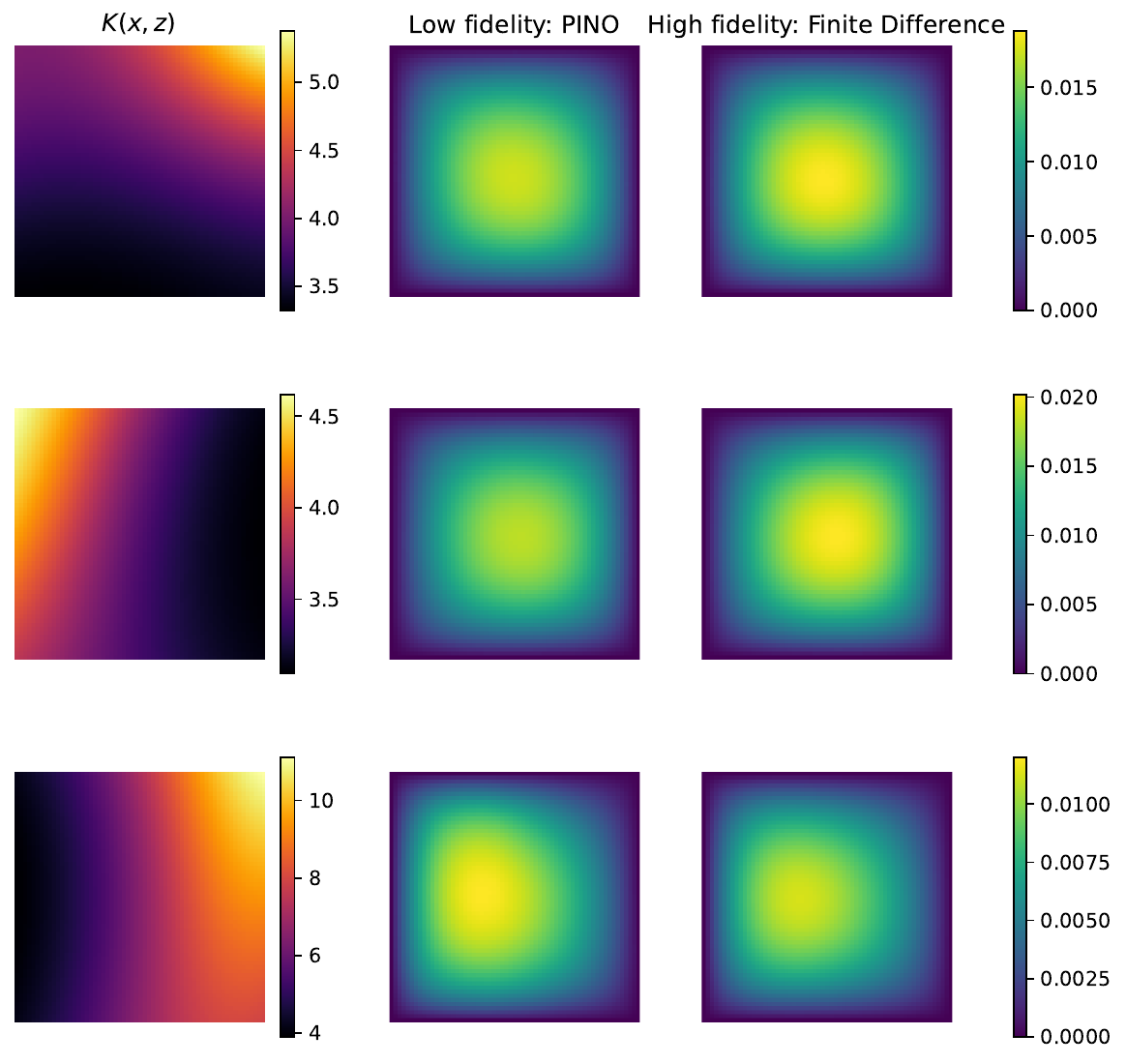}
        \subcaption{}
        \label{fig:darcy-examples}
    \end{minipage}
    \hfill
    \begin{minipage}[b]{0.49\textwidth}
        \centering
        \includegraphics[width=\textwidth]{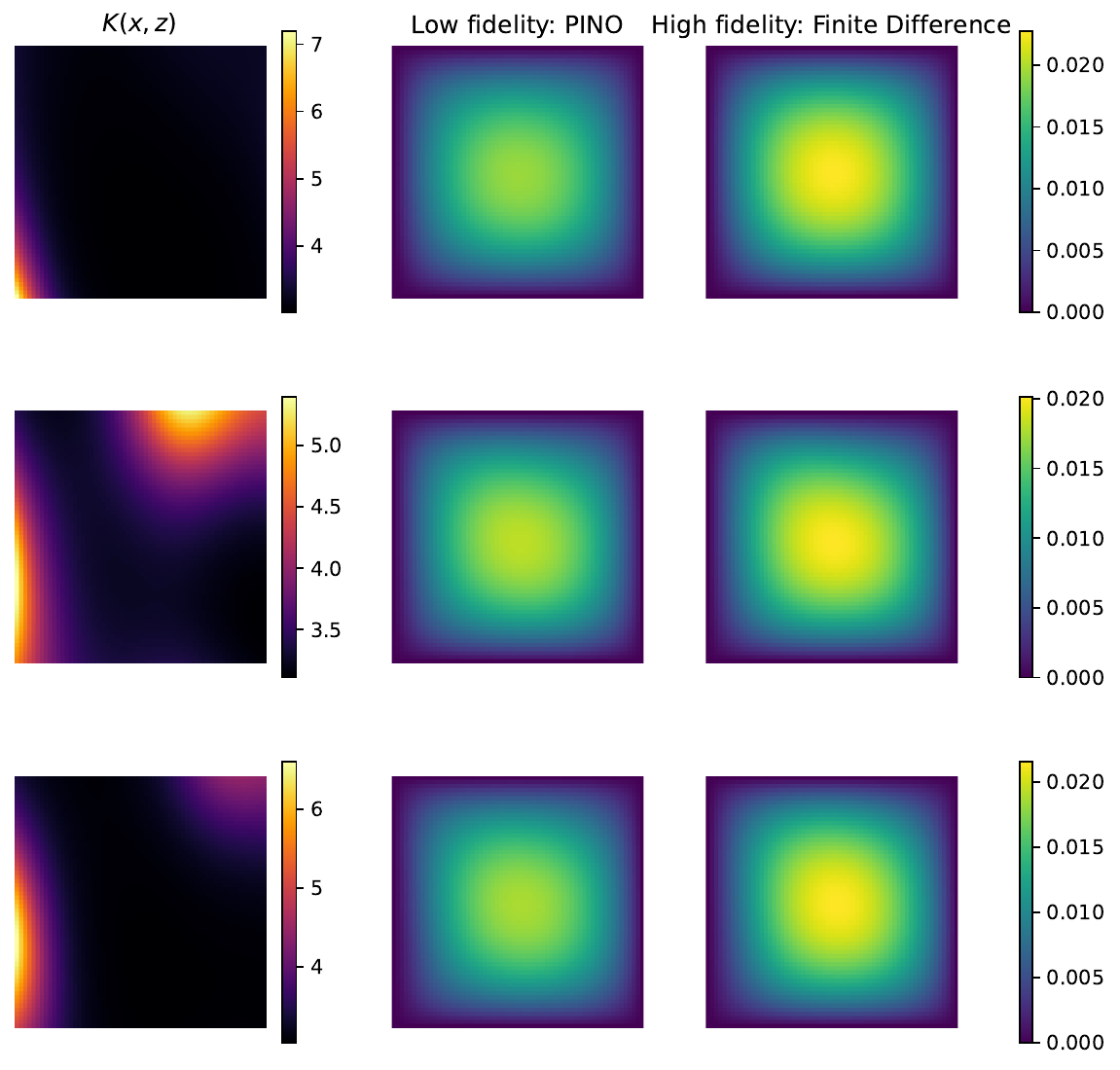}
        \subcaption{}
        \label{fig:darcy-q-examples}
    \end{minipage}
    \caption{The solutions of the steady-state heat equation in Equation~\eqref{eqn:steady-state} given three different realizations of the thermal coefficient $K(\bm x, \bm z)$ on a $61\times 61$ grid over $(0, 1)^2$ sampled from Equation~\eqref{eq:thermal-coef} (a) or $q(\bm z)$ (b). For both figures, the left column is the visualization of the thermal coefficient, the middle column is the LF QoI solution provided by a pre-trained PINO, and the right column is the HF QoI solution computed using the finite difference method.}
    \label{fig:darcy-combined}
\end{figure}

\begin{figure}[htp]
    \begin{minipage}{0.49\textwidth}
    \includegraphics[width=1.0\textwidth]{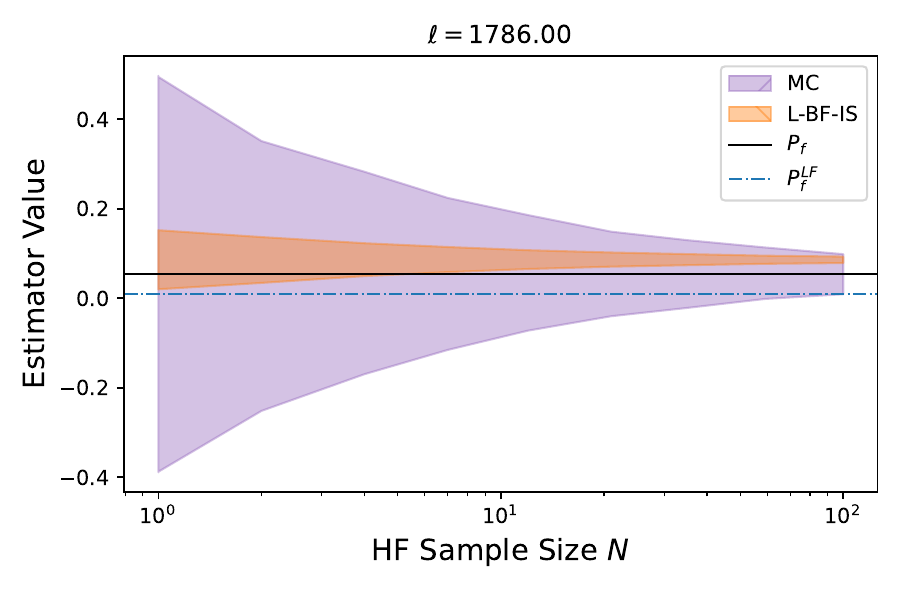}
    \subcaption{}
    \end{minipage}
    \hfill
    \begin{minipage}{0.49\textwidth}
    \includegraphics[width=1.0\textwidth]{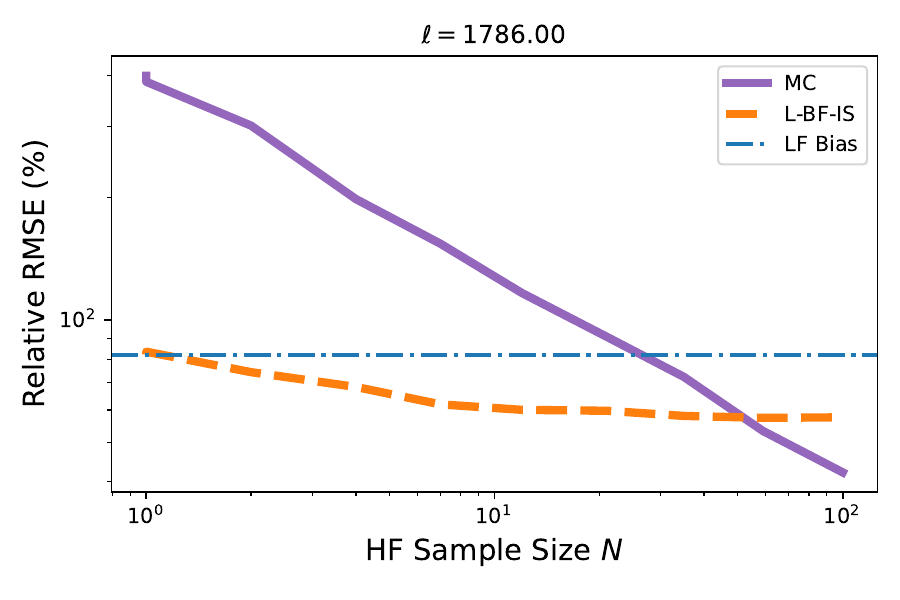}
    \subcaption{}
    \end{minipage}
    \caption{Convergence of the L-BF-IS (dashed) against standard Monte Carlo (solid) and LF failure probability (dashed dot) with $95\%$ confidence bound computed from $1,000$ trials for the steady-state heat equation problem in Section~\ref{sssec:heat}.}
    \label{fig:darcy}
\end{figure}

The convergence of the estimator is depicted in Figure~\ref{fig:darcy}. We restrict the performance comparison to cases where the number of HF evaluations is small ($\leq 63$). Notably, a non-trivial bias of approximately $2\%$ is observed. Despite demonstrating that the L-BF-IS estimator is unbiased in Section~\ref{ssec:estimator}, the samples produced by the Langevin algorithm, as described in Section~\ref{ssec:sampling}, cannot be guaranteed to represent $q(\bm z)$ exactly, thereby generating numerical bias in the observations. Given the problem's high dimensionality and suboptimal LF model, the L-BF-IS is able to reduce the relative RMSE from $85\%$ to $65\%$ using less than $N=100$ HF samples. Furthermore, the variance observed in the results generated by the L-BF-IS is substantially lower than that of standard Monte Carlo methods. This scenario underscores a case for employing pre-trained deep-learning-based operator learning strategies as an LF model. 

%%%%%%%%%%%%%%%%%%%%%%%%%%%%%%%%%%%%%%%%%%%%%%%%%%%%%%%%%%%%%%%%%%%%%%%%%% Conclusion
\section{Conclusion}
\label{sec:conclusion}
%%%%%%%%%%%%%%%
In this study, we present an importance sampling estimator, referred to as the Langevin bi-fidelity importance sampling (L-BF-IS). This estimator operates under the premise that in many practical applications, a considerably cheaper, differentiable lower-fidelity model is available. L-BF-IS employs the Metropolis-adjusted Langevin algorithm for sampling from a biasing distribution informed from the low-fidelity model, aiming to estimate failure probabilities with limited high-fidelity evaluations. The algorithm demonstrates superior performance in scenarios characterized by high input dimensions and multimodal failure regions. Two methodologies are introduced to tune a key parameter of the biasing distribution. Our empirical tests include a 1D manufactured bimodal function and two experimental setups using synthetic functions, with one involving 1000 random inputs. Additional experiments are conducted estimating failure probabilities of physics-based problems with failure probabilities of magnitude $1$-$5\%$. These experiments illustrate the efficiency of the L-BF-IS estimator relative to standard Monte Carlo simulation and a different importance sampling approach.

L-BF-IS demonstrates significant advantages, and our findings reveal opportunities to enhance the proposed estimator further. One promising direction is incorporating prior knowledge when selecting the Langevin algorithm's starting point and step size. Additionally, future research could explore treating the low-fidelity surrogate as dynamic, updating it at each iteration. This approach would extend the methodology into adaptive importance sampling, making it applicable when a fixed low-fidelity model is unavailable. These directions not only have the potential to refine the effectiveness of existing models but also pave the way for advancing state-of-the-art importance sampling techniques.

\section*{Acknowledgments} 
We would like to thank Grant Norman for the helpful discussion on the steady-state stochastic heat equation example.
N.\ Cheng and A.\ Doostan were supported by the AFOSR awards FA9550-20-1-0138 with Dr.\ Fariba Fahroo as the program manager, and by the US Department of Energy’s Wind Energy Technologies Office.

% \newpage 
% \bibliographystyle{abbrv}
% \bibliography{references}  %%% Uncomment this line and comment out the ``thebibliography'' section below to use the external .bib file (using bibtex) .

%% For citations use: 
%%       \citet{<label>} ==> Jones et al. [21]
%%       \citep{<label>} ==> [21]
%%

%% If you have bibdatabase file and want bibtex to generate the
%% bibitems, please use
%%
%%  \bibliographystyle{elsarticle-num-names} 
%%  \bibliography{<your bibdatabase>}

%% else use the following coding to input the bibitems directly in the
%% TeX file.
\bibliographystyle{elsarticle-num-names}
\bibliography{references}
% \begin{thebibliography}{00}

% %% \bibitem[Author(year)]{label}
% %% Text of bibliographic item

% \bibitem[ ()]{}

% \end{thebibliography}

%%%%%%%%%%%%%%%%%%%%%%%%%%%%%%%%%%%%%%%%%%%%%%%%%%%%%%%%%%%%%%%%%%%%%%%%%%%%%%%%%%%%%%%%%%%%%%%
% APPENDIX
%%%%%%%%%%%%%%%%%%%%%%%%%%%%%%%%%%%%%%%%%%%%%%%%%%%%%%%%%%%%%%%%%%%%%%%%%%%%%%%%%%%%%%%%%%%%%%%
\appendix
\section{Variance Deviation}
\label{apdx:variance}
The simplification of the $\widehat{P}^\BF_N$ variance is as follow:
\begin{equation}
\begin{aligned}
&\Var_{p\otimes q}\left[\widehat{P}^\BF_{M,N}\right]
\approx \frac{\mathcal{Z}^2(\ell)}{N}\Var_q\left[\mathbbm{1}_{h^\HF(\bm z)<0}\exp\left(\ell\tanh\circ h^\LF(\bm z) \right)\right]\\
&\quad= \frac{\mathcal{Z}^2(\ell)}{N}\left(\E_q\left[\mathbbm{1}_{h^\HF(\bm z)<0}\exp\left(2\ell\tanh\circ h^\LF(\bm z) \right)\right] - \left(\E_q\left[\mathbbm{1}_{h^\HF(\bm z)<0}\exp\left(\ell\tanh\circ h^\LF(\bm z) \right)\right]\right)^2\right)\\
&\quad= \frac{\mathcal{Z}^2(\ell)}{N}\left(\E_q\left[\mathbbm{1}_{h^\HF(\bm z)<0}\exp\left(2\ell\tanh\circ h^\LF(\bm z) \right)\right] - \left(\frac{1}{\mathcal{Z}(\ell)}\E_p\left[\mathbbm{1}_{h^\HF(\bm z)<0}\right]\right)^2\right)\\
&\quad= \frac{\mathcal{Z}(\ell)}{N} \E_p\left[\mathbbm{1}_{h^\HF(\bm z)<0}\exp\left(\ell\tanh\circ h^\LF(\bm z) \right)\right] - \frac{1}{N}\left(\E_p\left[\mathbbm{1}_{h^\HF(\bm z)<0}\right]\right)^2\\
&\quad= \frac{\mathcal{Z}(\ell)}{N} \E_p\left[\mathbbm{1}_{h^\HF(\bm z)<0}\exp\left(\ell\tanh\circ h^\LF(\bm z) \right)\right] - \frac{(\ExH)^2}{N}.
\end{aligned}
\end{equation}

\section{An Upper Bound for the Normalization Constant}
\begin{lemma}
\label{lm:cont-bound}
An upper bound for $\mathcal{Z}(\ell)$ is as
\begin{equation}
\mathcal{Z}(\ell) < (e^\ell - 1)\P_p[\AL] + 1.
\end{equation}
\end{lemma}
\begin{proof}
\begin{equation}
\begin{aligned}
\mathcal{Z}(\ell) &= \E_p[\exp\left(-\ell\tanh\circ h^\LF(\bm z)\right)]\\
&= \int_\Omega \exp\left(-\ell\tanh\circ h^\LF(\bm z)\right)p(\bm z)d\bm z\\
&= \int_{\AL} \exp\left(-\ell\tanh\circ h^\LF(\bm z)\right)p(\bm z)d\bm z + \int_{\AL^C} \exp\left(-\ell\tanh\circ h^\LF(\bm z)\right)p(\bm z)d\bm z\\
&< e^\ell\P_p[\AL] + \P_p[\AL^C]\\
&= (e^\ell - 1)\P_p[\AL] + 1.
\end{aligned}
\end{equation}
\end{proof}

\section{Simplification for KL Divergence}
\label{apdx:simp-kl}
The detailed process to simplify the KL divergence is,
\begin{equation}
\begin{aligned}
\KL(q^\ast\|q) &= \E_{q^\ast}\left[\log\frac{\mathcal{Z}(\ell)\mathbbm{1}_{h^\HF(\bm z) < 0}}{\ExH\exp\left(-\ell\tanh\circ h^\LF(\bm z)\right)}\right]\\
&= \log\frac{\mathcal{Z}(\ell)}{\ExH} + \E_{q^\ast}[\log \mathbbm{1}_{h^\HF(\bm z) < 0} + \ell\tanh\circ h^\LF(\bm z)]\\
&= \log\frac{\mathcal{Z}(\ell)}{\ExH} + \int_{\AH}\log \mathbbm{1}_{h^\HF(\bm z) < 0} + \ell\tanh\circ h^\LF(\bm z)d\bm z\\
&= \log\frac{\mathcal{Z}(\ell)}{\ExH} + \ell\int_{\AH\cap \AL}\tanh\circ h^\LF(\bm z)d\bm z + \ell\int_{\AH\cap \AL^C}\tanh\circ h^\LF(\bm z)d\bm z\\
&< \log\frac{(e^\ell - 1)\P_p[\AL] + 1}{\ExH} + \ell\P_p[\AH\cap \AL^C].
\end{aligned}
\end{equation}

\end{document}